\documentclass{article}

\usepackage{arxiv}

\usepackage[utf8]{inputenc} % allow utf-8 input
\usepackage[T1]{fontenc}    % use 8-bit T1 fonts
\usepackage{hyperref}       % hyperlinks
\usepackage{url}            % simple URL typesetting
\usepackage{booktabs}       % professional-quality tables
\usepackage{amsfonts}       % blackboard math symbols
\usepackage{nicefrac}       % compact symbols for 1/2, etc.
\usepackage{microtype}      % microtypography
\usepackage{graphicx}
\usepackage{doi}
\usepackage{amsthm}

\usepackage{notation}

\newtheorem{example}{Example}
\newtheorem{theorem}{Theorem}
\newtheorem{proposition}{Proposition}
\newtheorem{definition}{Definition}
\newtheorem{lemma}{Lemma}
\newtheorem{corollary}{Corollary}

\title{Finite two-dimensional proof systems for non-finitely axiomatizable logics}

\author{ \href{https://orcid.org/0000-0003-3240-386X}{\includegraphics[scale=0.06]{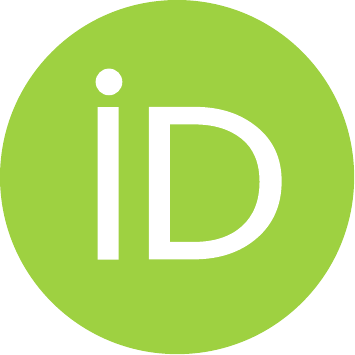}\hspace{1mm}Vitor Greati} \\
	Programa de Pós-graduação em Sistemas e Computação,\\
Departamento de Informática e Matemática Aplicada,\\
Universidade Federal do Rio Grande do Norte,\\ 
Natal, Brazil\\\\ 
Bernoulli Institute, University of Groningen, The Netherlands\\
	\texttt{vitor.greati.017@ufrn.edu.br} \\
	\And
	\href{https://orcid.org/0000-0003-2601-8164}{\includegraphics[scale=0.06]{orcid.pdf}\hspace{1mm}João Marcos} \\
	Departamento de Informática e Matemática Aplicada,\\
	Universidade Federal do Rio Grande do Norte,\\ 
	Natal, Brazil\\
	\texttt{jmarcos@dimap.ufrn.br} \\
}

\renewcommand{\shorttitle}{Finite two-dimensional proof systems for non-finitely axiomatizable logics}

\hypersetup{
pdftitle={A template for the arxiv style},
pdfsubject={q-bio.NC, q-bio.QM},
pdfauthor={David S.~Hippocampus, Elias D.~Striatum},
pdfkeywords={First keyword, Second keyword, More},
}

\begin{document}
\maketitle

\begin{abstract}
The characterizing properties of a proof-theoretical presentation of a given logic may hang on the choice of proof formalism, on the shape of the logical rules and of the sequents manipulated by a given proof system, on the underlying notion of consequence, and even on the expressiveness of its linguistic resources and on the logical framework into which it is embedded.  Standard (one-dimensional) logics determined by (non-deterministic) logical matrices are known to be axiomatizable by analytic and possibly finite proof systems as soon as they turn out to satisfy a certain constraint of sufficient expressiveness.  In this paper we introduce a recipe for cooking up a two-dimensional logical matrix (or \TheB-matrix) by the combination of two (possibly partial) non-deterministic logical matrices.  We will show that such a combination may result in \mbox{\TheB-matrices} satisfying the property of sufficient expressiveness, even when the input matrices are not sufficiently expressive in isolation, and we will use this result to show that one-dimensional logics that are not finitely axiomatizable may inhabit finitely axiomatizable two-dimensional logics, becoming, thus, finitely axiomatizable by the addition of an extra dimension.  
We will illustrate the said construction using a well-known logic of formal inconsistency called \mciName{}.  We will first prove that this logic is not finitely axiomatizable by a one-dimensional (generalized) Hilbert-style system.  Then, taking advantage of a known 5-valued non-deterministic logical matrix for this logic, we will combine it with another one, conveniently chosen so as to give rise to a \TheB-matrix that is axiomatized by a two-dimensional Hilbert-style system that is both finite and analytic.
\end{abstract}

% keywords can be removed

\section{Introduction}
\vspace{-2mm}
A {logic} is commonly defined nowadays
as a relation that connects collections of {formulas} 
from a {formal language} and satisfies some closure properties. 
The established connections are called {consecutions} and
each of them has two parts, an {antecedent} and a {succedent}, the latter often being said to `follow from' (or to be a consequence of) the former.
%A consecution aims at saying that the formulas in its succedent follow from (or are consequences of) the formulas in its antecedent.
A~logic may be manufactured in a number of ways, in particular
as being induced by the set of {derivations} 
justified by the {rules of inference} of a given
{proof system}. There are different kinds of proof systems,
the differences between them residing mainly in the shapes of their rules of inference
and on the way derivations are built.
We will be interested here in {Hilbert-style
proof systems} (`{H-systems}', for short), whose rules of inference have the same
shape of the consecutions of the logic they canonically induce
and whose associated derivations consist in expanding a given antecedent
by applications of rules of inference until the
desired succedent is produced.
A remarkable property of an H-system is that the logic induced by it is the least logic containing the rules of inference of the system;
in the words of~\cite{wojcicki1970}, the system constitutes a `{logical basis}' for the said logic.

%Traditionally, H-systems were constrained to having rules of inference containing a single formula in their succedent. 
%We call them `\SetFmla{} H-systems'.
Conventional H-systems, which we here dub `\SetFmla{} H-systems', do not allow for more than one formula in the succedents of the consecutions that they manipulate.
%Correspondingly, the consecutions of the logics induced by such systems were obliged to obey to this same constraint.
Since 
%the developments of Shoesmith \& Smiley in 
\cite{ss1978},
however, 
we have learned that the simple elimination of this restriction on H-systems ---that is, allowing for sets of formulas rather than single formulas in the succedents---
brings numerous advantages, among which we mention: \emph{modularity} (correspondence between rules of inference and properties satisfied by a semantical structure),
\emph{analyticity} (control over the resources demanded to produce a derivation), and the automatic generation of analytic proof systems for a wide class of logics
specified by sufficiently expressive non-deterministics semantics, with an associated straightforward proof-search procedure~\cite{marcelino19syn,marcelino19woll}.
Such generalized systems, here dubbed `\SetSet{} H-systems', %generalize \SetFmla{} H-systems and, as expected, 
induce logics whose
consecutions involve succedents consisting in a collection of formulas, intuitively understood as `alternative conclusions'.

An H-system $\mathcal{H}$ is said to be an \emph{axiomatization}
for a given logic $\mathcal{L}$ when the logic induced by $\mathcal{H}$
coincides with $\mathcal{L}$.
A desirable property for an axiomatization is \emph{finiteness},
namely the property of 
consisting on a finite collection of schematic axioms and rules of inference.
A logic having a finite axiomatization
is said to be `{finitely based}'.
In the literature, one may find examples of logics having a quite
simple, finite semantic presentation, being, in contrast,
not finitely based in terms of \SetFmla{} H-systems~\cite{pala1994}.
These very logics, however, when seen as companions of
logics with multiple formulas in the succedent,
turn out to be finitely based in terms of \SetSet{} H-systems~\cite{marcelino19syn}.
In other words, by updating the underlying proof-theoretical and the logical formalisms, we are able to obtain a finite axiomatization
for logics which in a more restricted setting could not be said to be finitely based.
We may compare the above mentioned movement to the common mathematical practice of %bringing up new concepts or 
adding dimensions in order to provide better insight on some phenomenon. 
%or more easily solve a problem. 
%
A well-known example of that is given by the Fundamental Theorem of Algebra, which provides an elegant solution to the problem of determining the roots of polynomials over a single variable, demanding only that real coefficients should be replaced by complex coefficients. 
%An example lies in the study of the roots of polynomials over a single variable. Without imaginary numbers, that is, sticking only with the reals, understanding a phenomenon such as a polynomial apparently not having a root, or even obtaining elegant and useful results such as the fundamental theorem of algebra would not be feasible. 
%Another example is the transportation of objects of interest to a higher-dimensional space. 
%A famous usage of this technique 
Another example, 
from Machine Learning, is the `kernel trick' employed in support vector machines:
by increasing the dimensionality of the input space, the transformed data points become more easily separable by hyperplanes, making it possible to achieve better results in classification tasks.%

\vspace{-.5mm}
%Back to the realm of logic, 
It is worth noting that there are logics that fail to be finitely based
in terms of \SetSet{} H-systems.
An example of a logic designed with the sole purpose of illustrating this possibility was provided in~\cite{marcelino19syn}.
One of the goals of the present work is to 
show that an important logic from the
literature of logics of formal inconsistency (LFIs)
called \mciName{} is also an example of this phenomenon.
This logic results from adding infinitely-many axiom
schemas to the logic \mbcName{}, 
a logic that is 
obtained by extending positive
classical logic with 
two axiom schemas.
%the axiom schemas $\PropA\lor\neg\PropA$ ({excluded middle}) and $\mciCons\PropA\supset(\mciNeg\PropA \supset(\PropA \supset \PropB))$ ({gentle explosion}).
Incidentally, along the proof of this result,
we will show that
\mciName{} is the limit of a strictly increasing
chain of LFIs
extending \mbcName{} (comparable to the case of 
%$\mathcal{C}_\omega$ 
$\mathcal{C}_{\mathrm{Lim}}$ 
in da Costa's hierarchy of increasingly weaker paraconsistent calculi~\cite{carniellimarcos1999}).
A natural question, then, is whether we can 
enrich our technology, in the same vein, in order to
provide finite axiomatizations for 
all these logics.
%
%We take a step towards an answer to that 
We answer that in the affirmative by means of the two-dimensional
frameworks
%% logical frameworks and proof systems
%notions of logics and H-systems 
developed in~\cite{blasiomarcos2017,greati2021}.
Logics, in this case, connect pairs of collections of formulas. 
%The first components of these pairs are intuitively interpreted as collections of `{accepted}' formulas, and the second components, as collections of `{rejected}' formulas.
%%A consecution in this setting may be read as saying that accepting one of the formulas in the first component or rejecting one formula in the second component of the succedent is a consequence of accepting all the formulas in the first component and rejecting all the formulas in the second component of the antecedent.
A consecution, in this setting, may be read as involving formulas that are accepted and those that are not, as well as formulas that are rejected and those that are not.
`Acceptance' and `rejection'
are seen, thus, as two orthogonal dimensions that may interact, making it possible, thus, to express more complex consecutions than those expressible in one-dimensional logics.
%Two-dimensional logics get even more interesting when we
%notice that we can identify many 
%different one-dimensional logics living inside them, 
%which may differ not only
%in terms of consecutions, but in the abstract properties that they satisfy~\cite{blasio20171}.
Two-dimensional H-systems, which we call `{\SetTSetT{} H-systems}',
%suitably 
generalize \SetSet{} H-systems so as to manipulate
pairs of collections of
%accepted and rejected 
formulas, canonically
inducing two-dimensional logics and constituting logical bases
for them.
Another goal of the present work is, therefore, to show how to obtain a two-dimensional
logic 
inhabited by a (possibly not finitely based)
one-dimensional logic of interest. More than that,
the logic we obtain
will be finitely axiomatizable in
terms of a \SetTSetT{} analytic H-system.
The only requirements is that the one-dimensional logic
of interest
must have an associated semantics in terms of
a finite non-deterministic logical matrix
and that this matrix can be 
combined with another one through a novel procedure that we 
will introduce, resulting in
%extended with
%a new set of designated elements to
a two-dimensional non-deterministic matrix (a \TheB-matrix~\cite{blasio20171})
satisfying a certain condition of sufficient expressiveness~\cite{greati2021}.
%This approach will be, in particular,  applied here to provide, to the best of our knowledge, the first finite and analytic axiomatization of \mciName{}.
An application of this approach will be provided here in order to produce the first finite and analytic axiomatization of~\mciName{}.

%We highlight below our main contributions:
%\begin{itemize}
%
%    %\item analytic one-dimensional 
%    %axiomatizations for
%    %sufficiently expressive logical matrices
%    %$\Struct{\ValSetA, \Acc}$ and $\Struct{\ValSetB, \Rej}$
%    %over, respectively, $\SigA_1$ and $\SigA_2$
%    %may be merged into a
%    %two-dimensional analytic axiomatization for
%    %the ($\SigA_1 \cup \SigA_2$)-\TheB-matrix $\Struct{\ValSetA \cup \ValSetB, \Acc, \Rej}$;
%    \item some non-sufficiently expressive logical matrices
%    may be extended to a sufficiently expressive
%    \TheB-matrix over the same algebra;
%    %only by introducing a new
%    %set of distinguished values;
%    \item one-dimensional logics that are not finitely
%    axiomatizable
%    may inhabit finitely axiomatizable two-dimensional logics,
%    being, thus, finitely axiomatizable in two dimensions;
%    and
%    \item \mciName{} is the limit of an 
%    infinite strictly increasing chain of 
%    LFIs extending \mbcName{},
%    and is an example of 
%    a logic that is finitely axiomatizable in 
%    two dimensions but not in one dimension.
%    %\item there are logics \emph{finitely} axiomatizable
%    %by a Hilbert-style system only in two-dimensions;
%    %\item \mciName{} is not \emph{finitely} axiomatizable 
%    %neither by a \SetFmla{} nor by a \SetSet{} Hilbert-style system;
%    %\item \mciName{} is \emph{finitely} axiomatizable by
%    % two-dimensional \emph{analytic} Hilbert-style system;
%\end{itemize}

\vspace{-.5mm}
The paper is organized as follows: Section~\ref{sec:prelims} introduces basic terminology and definitions regarding algebras and languages. Section~\ref{sec:one-dim} presents the notions of one-dimensional logics and \SetSet{} H-systems. Section~\ref{sec:mci-non-finit} proves that \mciName{} is not finitely axiomatizable by one-dimensional H-systems. Section~\ref{sec:two-dim} introduces %basic terminology and concepts concerning 
two-dimensional logics and H-systems, and describes the approach to extending a logical matrix to a \TheB-matrix with the goal of finding a finite two-dimensional axiomatization for the logic associated with the former. 
Section~\ref{sec:mci-two-d} presents a two-dimensional finite analytic H-system for \mciName{}. In the final remarks, we highlight some byproducts of our present approach and some features of the resulting proof systems, in addition to pointing to some directions for further research.%
\footnote{%
%The new axiomatization of \mciName{} and 
Detailed proofs of some results may be found 
% in the arXiv paper?
in accompanying appendices.%
}

%\newpage %!!
\section{Preliminaries}
\vspace{-2mm}
\label{sec:prelims}
%%% ALGEBRAS
A \emph{propositional signature}
is a family
$\SigA \SymbDef \{\SigArity{\Sigma}{k}\}_{\ArityA \in \NatSet}$,
where each $\SigArity{\Sigma}{\ArityA}$ is
a collection of $\ArityA$-ary \emph{connectives}.
We say that $\SigA$ \emph{is finite} when its base set $\bigcup_{k\in\NatSet}\SigArity{\SigA}{k}$ is finite.
%When $\bigcup_{k\in\NatSet}\SigArity{\SigA}{k}$ is finite, we say that $\SigA$ \emph{is finite}.
%
A \DefEmph{non-deterministic algebra over $\SigA$}, or simply \DefEmph{$\SigA$-nd-algebra}, is a structure
$\AlgA \SymbDef \Struct{\ValSetA, \AlgInterp{\FunPlaceholder}{\AlgA}}$,
such that $\ValSetA$ is a non-empty collection of values
called the \DefEmph{carrier} of $\AlgA$,
and, for each $\ArityA \in \NatSet$ and $\ConA \in \SigArity{\SigA}{\ArityA}$,
the multifunction
$\AlgInterp{\ConA}{\AlgA} : \ValSetA^\ArityA \to \PowerSet{\ValSetA}$
is the \DefEmph{interpretation of~$\ConA$ in~$\AlgA$}.
When $\SigA$ and $\ValSetA$ are finite, we say that $\AlgA$ is \DefEmph{finite}.
When the range of all interpretations of $\AlgA$ contains only
singletons, $\AlgA$ is said to be a \DefEmph{deterministic algebra over $\SigA$},
or simply a \DefEmph{$\SigA$-algebra}, meeting the usual definition
from Universal Algebra~\cite{burris1981}.
When $\EmptySet$ is not in the range
of each $\AlgInterp{\ConA}{\AlgA}$,
$\AlgA$ is said to be \DefEmph{total}.
% (as opposed to being \DefEmph{properly partial}).
Given a $\SigA$-algebra $\AlgA$ and a $\ConA \in \SigArity{\SigA}{1}$, we
let $\AlgInterp{\ConA}{\AlgA}^0(\ValA) \SymbDef \ValA$
and $\AlgInterp{\ConA}{\AlgA}^{i+1}(\ValA) \SymbDef \AlgInterp{\ConA}{\AlgA}(\AlgInterp{\ConA}{\AlgA}^i(\ValA))$.
%%%
A mapping $\ValuationA : \ValSetA \to \ValSetB$
is a \DefEmph{homomorphism} from $\AlgA$ to $\AlgB$ when,
for all $\ArityA \in \NatSet$, $\ConA \in \SigArity{\SigA}{\ArityA}$
and $\ValA_1,\ldots,\ValA_k \in \ValSetA$, we have
    $\ValAssignA[\AlgInterp{\ConA}{\AlgA}(\ValA_1,\ldots,\ValA_k)] \subseteq \AlgInterp{\ConA}{\AlgB}(\ValAssignA(\ValA_1),\ldots,\ValAssignA(\ValA_k))$.
%Notice that, in case $\AlgA$ is deterministic, 
%provided we identify singletons with their elements,
%we may simply
%write `$\in$' in the place of `$\subseteq$', and when both $\AlgA$
%and $\AlgB$ are deterministic, we may use `$=$' instead, matching
%thus the usual notion of homomorphism for $\SigA$-algebras.
The set of all homomorphisms from $\AlgA$ to $\AlgB$ is
denoted by $\HomSet{\SigA}{\AlgA}{\AlgB}$.
When $\AlgB = \AlgA$, we write $\EndSet{\SigA}{\AlgA}$, rather than $\HomSet{\SigA}{\AlgA}{\AlgA}$, for
the set of \DefEmph{endomorphisms on} $\AlgA$.

%%%% LANGUAGES
Let $\PropSetA$ be a denumerable collection of
$\DefEmph{propositional variables}$
and $\SigA$ be a propositional signature.
The absolutely free $\SigA$-algebra freely generated by
$\PropSetA$ is denoted by $\LangAlg{\SigA}{\PropSetA}$
and called the \DefEmph{$\SigA$-language generated by $\PropSetA$}.
The elements of $\LangSet{\SigA}{\PropSetA}$ are called
\DefEmph{$\SigA$-formulas}, and those among them that are
not propositional variables are called $\SigA$-\DefEmph{compounds}.
Given $\FmSetA \subseteq \LangSet{\SigA}{\PropSetA}$,
we denote by $\FmSetCompl{\FmSetA}$ the set $\SetDiff{\LangSet{\SigA}{\PropSetA}}{\FmSetA}$.
%When speaking informally, we may refer to
%$\SigA$-formulas simply as \DefEmph{sentences}
%or \DefEmph{propositions}.
The homomorphisms from $\LangAlg{\SigA}{\PropSetA}$ to $\AlgA$ are called \DefEmph{valuations on} $\AlgA$, 
and we denote by $\ValSet{\SigA}{\AlgA}$ the collection thereof.
%and we let
%$\ValSet{\AlgA} \SymbDef \HomSet{\LangAlg{\SigA}{\PropSetA}}{\AlgA}$.
Additionally, endomorphisms on $\LangAlg{\SigA}{\PropSetA}$ are dubbed
\DefEmph{$\SigA$-substitutions}, and we let $\SubsSet{\SigA}{\PropSetA} \SymbDef \EndSet{\SigA}{\LangAlg{\SigA}{\PropSetA}}$; 
when there is no risk of confusion, we may omit the superscript from this notation.
%When there is no risk of confusion, we may omit the set of propositional variables and simply write $\SubsSet{\SigA}{}$.
 
\sloppy 
Given $\FmA \in \LangSet{\SigA}{\PropSetA}$,
let 
%$\Subformulas{\FmA}$ be the set of \DefEmph{subformulas of} $\FmA$ 
%--- that is, $\Subformulas{\FmA} \SymbDef \{\PropA\}$
%if $\FmA = \PropA \in \PropSetA$
%and $\Subformulas{\ConA(\FmB_1,\ldots,\FmB_\ArityA)} \SymbDef %\{\FmA\} \cup \bigcup_{i=1}^{\ArityA}\Subformulas{\FmB_i}$ 
%if $\FmA = \ConA(\FmB_1,\ldots,\FmB_\ArityA)$ --- 
$\PropVars{\FmA}$ be the set of propositional
variables occurring in $\FmA$.
%--- that is, $\PropVars{\FmA} \SymbDef \Subformulas{\FmA} \cap \PropSetA$.
%The extensions of such operations to sets of $\SigA$-formulas
%are defined in the natural way.
If $\PropVars{\FmA} = \{\PropA_1,\ldots,\PropA_k\}$,
we say that $\FmA$ is $\ArityA$-ary (\emph{unary}, for $k=1$; \emph{binary}, for $k=2$) and %we
let
$\AlgInterp{\FmA}{\AlgA} : \ValSetA^\ArityA \to \PowerSet{\ValSetA}$
be \DefEmph{the $\ArityA$-ary multifunction on $\AlgA$ induced by~$\FmA$},
where, for all $\ValA_1,\ldots,\ValA_k \in \ValSetA$, we have
$\AlgInterp{\FmA}{\AlgA}(\ValA_1,\ldots,\ValA_k) \SymbDef
 \{\ValuationA(\FmA) \mid \ValuationA \in \ValSet{\SigA}{\AlgA} \text{ and } \ValuationA(\PropA_i) = \ValA_i, \text{ for } 1 \leq i \leq \ArityA\}$.
Moreover, given $\FmB_1,\ldots,\FmB_k \in \LangSet{\SigA}{\PropSetA}$, we write $\FmA(\FmB_1,\ldots,\FmB_k)$
for the $\SigA$-formula $\AlgInterp{\FmA}{\LangAlg{\SigA}{\PropSetA}}(\FmB_1,\ldots,\FmB_k)$,
and, where $\FmSetA\subseteq\LangSet{\SigA}{\PropSetA}$ is a set of
$\ArityA$-ary $\SigA$-formulas, we let $\FmSetA(\FmB_1,\ldots,\FmB_k) \SymbDef \{\FmA(\FmB_1,\ldots,\FmB_k) \mid \FmA \in \FmSetA\}$. 
%Another function on $\LangSet{\SigA}{\PropSetA}$
% The \DefEmph{size} of a $\SigA$-formula $\FmA$ is defined as the cardinality of the multiset of its subformulas, that is, $\FmSize{\FmA} \SymbDef 1$, if $\FmA \in \PropSetA$, and $\FmSize{\ConA(\FmB_1,\ldots,\FmB_k)} \SymbDef 1 + \sum_{i=1}^k\FmSize{\FmB_i}$ otherwise.
%
Given $\FmA \in \LangSet{\SigA}{\PropSetA}$,
by $\Subformulas{\FmA}$ we refer to the set of \DefEmph{subformulas of} $\FmA$. 
Where $\BAnaFormA$ is a unary $\SigA$-formula, we define the set $\GenSubfmlas{\BAnaFormA}{\FmA}$ 
%of \emph{$\BAnaFormA$-instantiated
%subformulas of $\FmA$} 
as $\{\SubstA(\BAnaFormA)\mid \SubstA~:~\PropSetA \to \Subformulas{\FmA}\}$.  Given a set $\BAnaSetA \supseteq \Set{\PropA}$ of unary $\SigA$-formulas, we set $\GenSubfmlas{\BAnaSetA}{\FmA} \SymbDef \bigcup_{\BAnaFormA\in\BAnaSetA}\GenSubfmlas{\BAnaFormA}{\FmA}$.
%Given $\BAnaSetA, \Set{\FmA} \subseteq \LangSet{\SigA}{\PropSetA}$ with $\PropVars{\FmB} = \Set{\PropA}$ for all $\FmB \in \BAnaSetA$, the set of \emph{$\BAnaSetA$-instantiated subformulas of $\FmA$} is defined as $\GenSubfmlas{\BAnaSetA}{\FmA} \SymbDef \Subformulas{\FmA} \cup \{\SubstA(\FmB) \mid \FmB \in \BAnaSetA, \SubstA : \PropSetA \to \Subformulas{\FmA}\}$.
%%that is, it is the set
%%of subformulas of $\FmA$ together with
%%all formulas resulting from placing
%%such subformulas in place of $\PropA$
%%in the formulas in $\BAnaSetA$.
For example, if $\BAnaSetA = \Set{\PropA,\mciNeg\PropA}$, we will have $\GenSubfmlas{\BAnaSetA}{\mciNeg(\PropB\lor\PropC)} = \Set{\PropB,\PropC,\PropB\lor\PropC,\mciNeg(\PropB\lor\PropC)}\cup\Set{ \mciNeg\PropB,\mciNeg\PropC,\mciNeg(\PropB\lor\PropC),\mciNeg\mciNeg(\PropB\lor\PropC)}$. Such generalized notion of %$\BAnaSetA$-instantiated
subformulas will be used in the next section %in a moment 
to provide a more generous proof-theoretical concept of \emph{analyticity}.

\section{One-dimensional consequence relations}
\label{sec:one-dim}
A \DefEmph{\SetSet{} statement} (or \DefEmph{sequent}) is a pair $\Pair{\FmSetA}{\FmSetB} \in \PowerSet{\LangSet{\Sigma}{\PropSetA}} \times \PowerSet{\LangSet{\Sigma}{\PropSetA}}$, 
where $\FmSetA$ is dubbed the \emph{antecedent} and $\FmSetB$ the \emph{succedent}.  A \DefEmph{one-dimensional consequence relation on $\LangSet{\SigA}{\PropSetA}$} is a collection~$\SetSetCR{}$ of \SetSet{} statements
satisfying, for all $\FmSetA,\FmSetB,\FmSetA^\prime,\FmSetB^\prime \subseteq \LangSet{\SigA}{\PropSetA}$,
\begin{description}[labelindent=0.5cm, labelwidth=1cm]
    \item[\namedlabel{prop:CRSSPropO}{\CRSSPropO}] if $\FmSetA \cap \FmSetB \neq \EmptySet$, then $\FmSetA \SetSetCR{} \FmSetB$
    \item[\namedlabel{prop:CRSSPropD}{\CRSSPropD}] if $\FmSetA \SetSetCR{} \FmSetB$, then
    $\FmSetA\cup\FmSetA^\prime \SetSetCR{}{} \FmSetB\cup\FmSetB^\prime$
    \item[\namedlabel{prop:CRSSPropC}{\CRSSPropC}] if
    $\CutSetSS\cup\FmSetA \SetSetCR{} \FmSetB\cup\FmSetCompl{\CutSetSS}$
    for all $\CutSetSS\subseteq\LangSet{\SigA}{\PropSetA}$,
    then $\FmSetA \SetSetCR{} \FmSetB$
\end{description}
Properties \ref{prop:CRSSPropO}, \ref{prop:CRSSPropD}
and \ref{prop:CRSSPropC} are called
\DefEmph{overlap}, \DefEmph{dilution}
and \DefEmph{cut}, respectively. 
%\SetSet{} consequence relations are often referred to as \emph{Scott consequence relations}, in view of the writings of Dana Scott on this subject~\cite{scott1974} (even though Scott requires the relation to hold between finite sets only).
%\TODO{Talk a little bit more about (C), based on SS}
%Equivalent, and sometimes more useful, versions
%of \ref{prop:CRSSPropC} are
%\begin{description}[labelindent=1cm, labelwidth=1cm]
%    \item[\namedlabel{prop:CRSSPropCF}{\CRSSPropCF}] if, for some
%    $\FmSetC \subseteq \LangSet{\SigA}{\PropSetA}$,
%    we have
%    $\FmSetD, \FmSetA \SetSetCR{}{} \FmSetB, \FmSetCompl{\FmSetD}$
%    for all $\FmSetD \subseteq \FmSetC$,
%    then $\FmSetA \SetSetCR{}{} \FmSetB$
%    \item[\namedlabel{prop:CRSSPropCL}{\CRSSPropCL}]
%\end{description}
The relation $\SetSetCR{}$ is called \DefEmph{substitution-invariant} when it satisfies, for every $\SubstA \in \SubsSet{\SigA}{}$, 
\begin{description}[labelindent=0.5cm, labelwidth=1cm]
    \item[\namedlabel{prop:CRSSPropSS}{\CRSSPropSS}]
    if $\FmSetA \SetSetCR{} \FmSetB$,
    then $\sigma[\FmSetA] \SetSetCR{} \sigma[\FmSetB]$
\end{description}
and it is called \DefEmph{finitary} when it satisfies
\begin{description}[labelindent=0.5cm, labelwidth=1cm]
    \item[\namedlabel{prop:CRSSPropF}{\CRSSPropF}] 
    if $\FmSetA \SetSetCR{} \FmSetB$,
    then 
    $\FinSet{\FmSetA} \SetSetCR{} \FinSet{\FmSetB}$
    for some finite
    $\FinSet{\FmSetA} \subseteq \FmSetA$ and
    $\FinSet{\FmSetB} \subseteq \FmSetB$
\end{description}
One-dimensional consequence
relations will also be referred
to as \DefEmph{one-dimen\-sion\-al logics}.
Substitution-invariant finitary one-dimensional logics
will be called \DefEmph{standard}.
We will denote by $\nSetSetRel{}{}$ the complement
    of $\SetSetRel{}{}$, called the \DefEmph{compatibility relation associated with} $\SetSetRel{}{}$~\cite{blasio2021}.
    
    % When we restrict the above notion of (substitution-invariant, finitary) consequence relation to allow for single formulas in the succedent of its statements, we obtain the notion of (substitution-invariant, finitary) \DefEmph{Tarskian consequence relations}.
    %, which we call here
    %\SetFmla{} consequence relations.
    A \DefEmph{\SetFmla{} statement} is a sequent having a single formula as consequent.
    When we restrict standard consequence relations to collections of \SetFmla{} statements, we define the so-called (substitution-invariant finitary) \DefEmph{Tarskian consequence relations}.
Every one-dimensional consequence relation
    $\SetSetRel{}{}$ determines a Tarskian consequence relation $\SetFmlaRel{\SetSetRel{}{}}{} \;\subseteq\PowerSet{\LangSet{\SigA}{\PropSetA}} \times \LangSet{\SigA}{\PropSetA}$,
    dubbed \DefEmph{the \SetFmla{} Tarskian companion of $\SetSetRel{}{}$},
    such that, for all $\FmSetA \cup \{\FmB\} \subseteq \LangSet{\SigA}{\PropSetA}$,
    $\FmSetA \SetFmlaRel{\SetSetRel{}{}}{} \FmB$ if, and only if,
    $\FmSetA \SetSetRel{}{} \{\FmB\}$.
%
    %In some situations,
    %knowing that a \SetFmla{} consequence relation is the %\SetFmla{} companion of
    %a \SetSet{} consequence relation enables us to draw useful %conclusions
    %about each of them (see~\cite{ss1978}).
    It is well-known that the 
    collection of all
    Tarskian consequence relations
    over a fixed language
    constitutes a complete lattice under 
    set-theoretical inclusion~\cite{wojcicki1988}. Given
    a set~$C$ of such
    relations, we will denote by $\Supremum C$
    its supremum in the latter lattice.
    
We present in what follows two ways of obtaining one-dimensional consequence relations:
one semantical, 
%(denoted by $\SetSetMCR{\MatA}$), 
via 
non-deterministic logical matrices~\cite{avron2011},
and the other proof-theoretical, 
%(denoted by $\SSHCR{\SSHCalcA}$), 
via \SetSet{} Hilbert-style systems~\cite{ss1978,marcelino19syn}.
    
A \DefEmph{non-deterministic $\SigA$-matrix}, or simply \DefEmph{$\SigA$-nd-matrix}, 
is a structure $\MatA \SymbDef \Struct{\AlgA,\DesSetA}$, 
where $\AlgA$ is a $\SigA$-nd-algebra,
whose carrier is the set of \DefEmph{truth-values},
and $\DesSetA \subseteq \ValSetA$ is the set of \DefEmph{designated truth-values}. 
Such structures are also known in the literature
as `PNmatrices'~\cite{baaz2013};
they generalize the so-called `Nmatrices'~\cite{avron2005structures}, which 
%are $\SigA$-nd-matrices with the restriction that $\AlgA$ must be total by allowing $\AlgA$ to be properly partial.
are $\SigA$-nd-matrices with the restriction that $\AlgA$ must be total.
%When $\AlgA$ is deterministic, $\MatA$ is \DefEmph{deterministic}, and when $\AlgA$ is total, 
%$\MatA$ is also said to be \DefEmph{total}.
From now on, whenever $\ValSubsetA \subseteq \ValSetA$, we denote
$\SetDiff{\ValSetA}{\ValSubsetA}$ by $\ValueSetComp{\ValSubsetA}$.
In case $\AlgA$ is deterministic, we simply say that $\MatA$
is a \DefEmph{$\SigA$-matrix}.
Also, $\MatA$ is said to be \DefEmph{finite} when $\AlgA$ is finite. 
%When talking about $\MatA$, we may write $\ValSetMatrix{\MatA}$
%to refer to $\ValSet{\AlgA}$, the set of \DefEmph{$\MatA$-valuations}.
Every $\SigA$-nd-matrix $\MatA$ determines a substitution-invariant one-dimensional consequence
relation over $\SigA$, denoted by $\SetSetMCR{\MatA}$,
such that
%\begin{equation}
\begin{math}
    \FmSetA \SetSetMCR{\MatA} \FmSetB
    \text{ if, and only if, for all }
    \ValuationA\in\ValSetMatrix{\SigA}{\AlgA},
    \ValuationA[\FmSetA]~\cap~\ValuesSetCompl{\DesSetA} \neq \EmptySet
    \text{ or }
    \ValuationA[\FmSetB] \cap  \DesSetA \neq \EmptySet.
\end{math}
%\end{equation}
\noindent It is worth noting that
$\SetSetMCR{\MatA}$ is
finitary whenever the carrier of $\AlgA$ is finite (the proof runs very similar to that of the same result for Nmatrices~\cite[Theorem 3.15]{avron2005structures}).

A \DefEmph{strong homomorphism} between $\SigA$-matrices 
$\MatA_1 \SymbDef \Struct{\AlgA_1,D_1}$
and $\MatA_2 \SymbDef \Struct{\AlgA_2, D_2}$
is a homomorphism $h$ between $\AlgA_1$ and $\AlgA_2$
such that $\ValA \in D_1$ if, and only if, $h(\ValA) \in D_2$.
When there is a surjective strong homomorphism between $\MatA_1$
and $\MatA_2$, we have that $\SetSetCR{\MatA_1} = \SetSetCR{\MatA_2}$.

Now, to the Hilbert-style systems.
A (\DefEmph{schematic}) \DefEmph{\SetSet{} rule of inference} $\SSHRuleOfInf_\SSHRuleSchema$ 
is the collection
of all substitution instances
of the \SetSet{} statement $\SSHRuleSchema$,
called the \emph{schema} of $\SSHRuleOfInf_\SSHRuleSchema$.
%$\Pair{\FmSetC}{\FmSetD} \in \PowerSet{\LangSet{\SigA}{\PropSetA}}\times\PowerSet{\LangSet{\SigA}{\PropSetA}}$, called rule instances and usually denoted by $\HRule{\FmSetC}{\FmSetD}$,
Each $\SSHRuleInst \in  \SSHRuleOfInf_\SSHRuleSchema$
is called a \emph{rule instance of $\SSHRuleOfInf_\SSHRuleSchema$}.
%where $\FmSetC$ is the \DefEmph{antecedent} and
%$\FmSetD$ is the \DefEmph{succedent}
%of $\SSHRuleInst$.
A (\DefEmph{schematic}) \DefEmph{\SetSet{} H-system} $\SSHCalcA$ is a collection of
\SetSet{} rules of inference.
%We denote by $\SSHInsts{\SSHCalcA}$
%the union of all rules of inference of %$\SSHCalcA$,
%that is, all the rule instances of $\SSHCalcA$.
When we constrain the rule instances of $\SSHCalcA$
to having only singletons as succedents, 
we obtain the conventional notion of Hilbert-style system,
called here \emph{\SetFmla{} H-system}.

An \emph{$\SSHCalcA$-derivation} 
in a \SetSet{} H-system $\SSHCalcA$
is a rooted directed
tree $\SSHTreeA$ 
such that every node is 
labelled with sets of formulas or with
a discontinuation symbol~$\StarLabel$, and in which
every non-leaf node (that is, a node with child nodes) $\SSHNodeA$ in $\TreeA$ is an \emph{expansion of $\SSHNodeA$
by a rule instance} $\SSHRuleInst$ of $\SSHCalcA$. This means that
the antecedent of $\SSHRuleInst$ is contained in the label of $\SSHNodeA$
%--- that is, $\SSHNodeA$ \emph{satisfies}
%the antecedent of $\SSHRuleInst$ ---
and that
$\SSHNodeA$ has exactly one child node for
each formula $\FmB$ in the succedent of~$\SSHRuleInst$.
These child nodes are, in turn, labelled with the same formulas as those of $\SSHNodeA$
plus the respective formula $\FmB$. In case $\SSHRuleInst$ has an empty
succedent, then $\SSHNodeA$ has a single child node
labelled with $\StarLabel$. 
Here we will consider only \emph{finitary} \SetSet{} 
H-systems, in which each rule instance has
finite antecedent and succedent.
In such cases, we only need to consider finite
derivations.
Figure~\ref{fig:derivationscheme}
illustrates how derivations 
using only finitary rules of inference may be graphically represented.
We denote by $\LabelFn{\SSHTreeA}(\SSHNodeA)$
the label of the node $\SSHNodeA$ in the tree
$\SSHTreeA$.
It is worth observing that, for \SetFmla{} H-systems, derivations are
linear trees (as rule instances have a single formula in their succedents), or, in other words,
just sequences of formulas built by applications of
the rule instances, matching thus
the conventional definition of Hilbert-style 
systems.
%%%%%%%%%%%%%%%%%%
\vspace{-1em}
%%%%%%%%%%%%%%%%%%
%
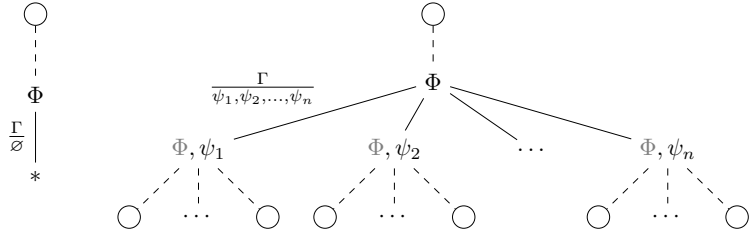
\begin{figure}[htb]
    \centering
    \scalebox{0.9}{
    \begin{tikzpicture}[every tree node/.style={},
       level distance=1.2cm,sibling distance=1cm,
       edge from parent path={(\tikzparentnode) -- (\tikzchildnode)}, baseline]
        \Tree[.\node[style={draw,circle}] {};
            \edge[dashed];
            [.\node[style={}] {$\FmSetA$};
            \edge node[auto=right] {$\HRule{\FmSetC}{\EmptySet}$};
            [.{$\StarLabel$}
            ]
            ]
        ]
    \end{tikzpicture}}
    \qquad
    \scalebox{0.9}{
    \begin{tikzpicture}[every tree node/.style={},
       level distance=1cm,sibling distance=.5cm,
       edge from parent path={(\tikzparentnode) -- (\tikzchildnode)}, baseline]
        \Tree[.\node[style={draw,circle}] {};
            \edge[dashed];
            [.\node[style={}] {$\FmSetA$};
            \edge node[auto=right] {$\HRule{\FmSetC}{\FmB_1,\FmB_2,\ldots,\FmB_n}$};
            [.${\color{gray}\FmSetA},\FmB_1$
                \edge[dashed];
                [.\node[style={draw,circle}]{};
                ]
                \edge[dashed];
                [.\node[style={}]{$\cdots$};
                ]
                \edge[dashed];
                [.\node[style={draw,circle}]{};
                ]
            ]
            [.${\color{gray}\FmSetA},\FmB_2$
                \edge[dashed];
                [.\node[style={draw,circle}]{};
                ]
                \edge[dashed];
                [.\node[style={}]{$\cdots$};
                ]
                \edge[dashed];
                [.\node[style={draw,circle}]{};
                ]
            ]
            [.$\ldots$
            ]
            [.${\color{gray}\FmSetA},\FmB_n$
                \edge[dashed];
                [.\node[style={draw,circle}]{};
                ]
                \edge[dashed];
                [.\node[style={}]{$\cdots$};
                ]
                \edge[dashed];
                [.\node[style={draw,circle}]{};
                ]
            ]
            ]
        ]
    \end{tikzpicture}
    }
    \caption{Graphical representation of $\SSHCalcA$-derivations, for
    $\SSHCalcA$ finitary.
    The dashed edges and blank circles represent other branches that may exist in the derivation.
    %in the case of expanded
    %nodes,
    We usually omit the formulas inherited from the parent
    node, 
    exhibiting only the ones introduced by the
    applied rule of inference.
    In both cases,
    we must have $\FmSetC \subseteq \FmSetA$
    to enable the application of the rule.}
    \label{fig:derivationscheme}
\end{figure}

%%%%%%%%%%%%%%%%
\vspace{-1.8em}
%%%%%%%%%%%%%%%%

A node $\SSHNodeA$ of an $\SSHCalcA$-derivation $\SSHTreeA$ is called
\DefEmph{$\FmSetD$-closed} in case 
it is a leaf node with
$\LabelFn{\SSHTreeA}(\SSHNodeA) = \StarLabel$
or
$\LabelFn{\SSHTreeA}(\SSHNodeA) \cap \FmSetD \neq \EmptySet$. 
A branch of $\SSHTreeA$ is $\FmSetD$-closed
when it ends in a $\FmSetD$-closed node.
When every branch in
$\SSHTreeA$ is $\FmSetD$-closed,
we say that $\SSHCalcA$ is itself
\DefEmph{$\FmSetD$-closed}.
	An \DefEmph{$\SSHCalcA$-proof} of a \SetSet{} statement
$\Statement{\FmSetA}{\FmSetB}$ is 
a $\FmSetB$-closed $\SSHCalcA$-derivation $\SSHTreeA$
such that $\LabelFn{\SSHTreeA}(\TreeRoot{\SSHTreeA}) \subseteq \FmSetA$.

Consider the binary relation $\SSHCR{\SSHCalcA}$ on $\PowerSet{\LangSet{\SigA}{\PropSetA}}$
 such that
$\FmSetA \SSHCR{\SSHCalcA} \FmSetB$
if, and only if, there is
an $\SSHCalcA$-proof of $\Statement{\FmSetA}{\FmSetB}$.
This relation is the
smallest substitution-invariant one-dimensional consequence relation
containing the rules of inference of $\SSHCalcA$, 
and it is finitary when $\SSHCalcA$ is finitary.
Since \SetSet{} (and \SetFmla{}) H-systems canonically
induce one-dimensional consequence relations,
we may refer to them as 
\DefEmph{one-dimensional H-systems} 
or \DefEmph{one-dimensional axiomatizations}.
In case there is a proof of 
$\Statement{\FmSetA}{\FmSetB}$
whose nodes are labelled only with
subsets of $\GenSubfmlasFn{\BAnaSetA}[\FmSetA \cup \FmSetB]$, we write $\FmSetA \SetSetRel{\SSHCalcA}{\BAnaSetA} \FmSetB$
%and say that this proof is %\emph{$\BAnaSetA$-analytic}
.
In case $\SSHCR{\SSHCalcA} = \SetSetRel{\SSHCalcA}{\BAnaSetA}$,
we say that $\SSHCalcA$ is
\emph{$\BAnaSetA$-analytic}.
Note that the ordinary notion of analyticity obtains when
$\BAnaSetA = \Set{\PropA}$.
From now on, whenever we use
the word ``analytic'' we will mean
this extended notion of
$\BAnaSetA$-analyticity,
for some $\BAnaSetA$ implicit in the context. When the
$\BAnaSetA$ happens to be important for us or
we identify any risk of confusion,
we will mention it explicitly.

In~\cite{marcelino19woll}, based on the seminal results on axiomatizability via \SetSet{} H-systems by Shoesmith and Smiley~\cite{ss1978}, it was proved that
any non-deterministic logical matrix $\MatA$ satisfying a criterion of sufficient expressiveness
is axiomatizable by a $\BAnaSetA$-analytic \SetSet{} Hilbert-style system,
which is finite whenever~$\MatA$ is finite,
where $\BAnaSetA$ is the set of separators
for the pairs of truth-values of $\MatA$.
According to such criterion, an nd-matrix is \emph{sufficiently expressive} when,
for every pair $\Pair{\ValA}{\ValB}$
of distinct truth-values,
there is a unary formula $\SepA$,
called a \emph{separator for $\Pair{x}{y}$},
such that $\AlgInterp{\SepA}{\AlgA}(\ValA) \subseteq \DesSetA$
and $\AlgInterp{\SepA}{\AlgA}(\ValB) \subseteq \ValuesSetCompl{\DesSetA}$, 
or vice-versa; 
in other words, when every pair of
distinct truth-values is \emph{separable in $\MatA$}.
%When this happens, the matrix can be
%fully described in terms of \SetSet{} schemas
%involving the separators (cf. \cite[Theorem X]{marcelino19woll}).

We emphasize that it is essential for the above result 
the adoption of \SetSet{} H-systems, instead
of the more restricted \SetFmla{} H-systems.
%%there are sufficiently expressive nd-matrices that are not finitely axiomatizable by \SetFmla{} H-systems, as is witnessed by the three-valued logics presented in~\cite{pala1994} (with deterministic two-valued matrices this cannot happen~\cite{rautenberg1981}), and by a more recent example in~\cite{marcelino2021}, with a simple two-valued non-deterministic matrix.
%
In fact, while two-valued matrices may always be finitely axiomatized by \SetFmla{} H-systems \cite{rautenberg1981}, there are sufficiently expressive three-valued deterministic matrices \cite{pala1994} and even quite simple two-valued non-deterministic matrices \cite{marcelino2021} that fail to be finitely axiomatized by \SetFmla{} H-systems.
When the nd-matrix at hand is not sufficiently expressive, we may observe the same phenomenon 
of not having a finite axiomatization also in terms
of \SetSet{} H-systems, even if the said nd-matrix
is finite. 
The first example (and, to the best of our knowledge, the only one in the current literature) of this fact appeared in~\cite{marcelino19woll}, which we reproduce here for later reference:

\begin{example}
\label{ex:three-valued-non-ax-set-set}
Consider the signature $\SigA\SymbDef \{\SigArity{\Sigma}{k}\}_{\ArityA \in \NatSet}$ such that 
$\SigA_1 \SymbDef \Set{g, h}$
and $\SigA_k \SymbDef \EmptySet$ for all
$k \neq 1$.
Let $\MatA \SymbDef \Struct{\AlgA, \Set{\tVal}}$ be a 
$\SigA$-nd-matrix, with $\ValSetA \SymbDef \Set{\tVal,\fVal,\BotVal}$ and 

\begin{align*}
\AlgInterp{g}{\AlgA}(\ValA)
=
\begin{cases}
   %\ValSetA, & \text{if } \ValA \neq \BotVal\\ 
   %\fVal, & \text{otherwise}\\ 
  \Set{\tVal}, & \text{if } \ValA = \BotVal\\ 
   \ValSetA, & \text{otherwise}\\
\end{cases}
\quad
\AlgInterp{h}{\AlgA}(\ValA)
=
\begin{cases}
   %\ValSetA, & \text{if } \ValA \neq \tVal\\ 
   %\tVal, & \text{otherwise}\\ 
   \Set{\fVal}, & \text{if } \ValA = \fVal\\ 
   \ValSetA, & \text{otherwise}\\
\end{cases}
\end{align*}
%\begin{table}[H]
%    \centering
%    \begin{tabular}{c|c|c}
%          $\ValA$& $\AlgInterp{f}{\AlgA}$& $\AlgInterp{g}{\AlgA}$\\
%         \midrule
%         \fVal& $\ValSetA$&$\ValSetA$\\
%         \BotVal& \fVal & $\ValSetA$\\
%         \tVal& $\ValSetA$ & \tVal
%    \end{tabular}
%\end{table}
\noindent This matrix is not sufficiently expressive
because there is no separator for the pair $\Pair{\fVal}{\BotVal}$, and
\cite{marcelino19woll} proved that
it is not axiomatizable by a finite \SetSet{} H-system,
even though an infinite \SetSet{} system that captures it
has a quite simple description in terms of the following infinite collection of schemas:
\vspace{-1mm}
\begin{gather*}
    \HRule{h^i(\PropA)}{\PropA, g(\PropA)},
    \text{ for all $i \in \NatSet$}.
\end{gather*}
\end{example}
\vspace{-1mm}

In the next section, we reveal another example of
this same phenomenon, this time
of the known LFI~\cite{carnielli2002}
called \mciName{}.
In the path of proving that this logic is
not axiomatizable by a finite \SetSet{}
H-system, we will show that
there are infinitely
many LFIs between \mbcName{} and \mciName{},
organized in a strictly increasing chain whose
limit is \mciName{} itself.

Before continuing, it is worth emphasizing that any given 
non-sufficiently expressive
nd-matrix may be conservatively extended to a
sufficiently expressive nd-matrix provided new connectives
are added to the language~\cite{marcelino19syn}. These new connectives have the sole purpose of
separating the pairs of truth-values for which no
separator is available in the original language.
The \SetSet{} system produced from this extended nd-matrix
can, then, be used to reason over the original logic, since
the extension is conservative. 
However, these new connectives,
which a priori have no meaning, are very likely to appear in
derivations of consecutions of the original logic.
%
%% For those who care about the language being used and the meaning of the involved logical constants, this is not an attractive option.
This might not look like an attractive option to inferentialists who believe that purity of the schematic rules governing a given logical constant is essential for the meaning of the latter to be coherently fixed.
In the subsequent sections, we will introduce and apply a potentially more expressive
notion of logic in order to
provide a \emph{finite} and \emph{analytic} H-system
for logics that are not finitely axiomatizable in one dimension,
while preserving their original languages.

\section{The logic \mciName{} is not finitely axiomatizable}
% in one dimension}
\label{sec:mci-non-finit}
A one-dimensional logic $\SSHCR{}$ over $\SigA$ is said to be
\mbox{\emph{$\mciNeg$-paraconsistent}} when
%, for some $\DefProp \in P$, $\nSetSetRel{}{}$
we have $\PropA,\mciNeg\PropA \;\nSetSetRel{}{}\; \PropB$,
for $\PropA,\PropB \in \PropSetA$.
%and \emph{$\mciNeg$-paracomplete} when
%%, for some $\DefProp \in P$, 
%we have $\PropB\;\nSetSetRel{}{}\;\PropA,\mciNeg\PropA$, with $\PropA,\PropB \in \PropSetA$.
%
Moreover, $\SSHCR{}$
is \emph{$\mciNeg$-gently explosive} in case there is a collection $\bigcirc(\PropA) \subseteq \LangSet{\SigA}{\PropSetA}$ of unary formulas such that, 
for some $\FmA \in \LangSet{\SigA}{\PropSetA}$,
we have
$\bigcirc(\FmA), \FmA \nSetSetRel{}{} \varnothing$;
$\bigcirc(\FmA), \mciNeg\FmA \nSetSetRel{}{} \varnothing$,
and, for all $\FmA \in \LangSet{\Sigma}{\PropSetA}$,
$\bigcirc(\FmA),\FmA,\mciNeg\FmA \; \SSHCR{} \varnothing$.
We say that $\SSHCR{}$ is a
\emph{logic of formal inconsistency (LFI)} in case it is $\mciNeg$-paraconsistent yet
$\mciNeg$-gently explosive.
In case $\bigcirc(\PropA) = \{\mciCons\PropA\}$,
for $\mciCons$ a (primitive or composite) \emph{consistency connective}, the logic is said also to be a \emph{\textbf{C}-system}.
In what follows,
let $\SigMCI$ be the propositional signature such that
$\SigMCI_{1} \SymbDef \{\mciNeg,\mciCons\}$,
$\SigMCI_{2} \SymbDef \{\land,\lor,\mciImp\}$,
and $\SigMCI_{k} \SymbDef \EmptySet$
for all $k \not\in \{1,2\}$.

One of the simplest
\textbf{C}-systems
is the logic \mbcName{}, which was 
first presented in terms of
a \SetFmla{} H-system over $\SigMCI$
obtained by extending 
any \SetFmla{} H-system for
positive classical logic (\CPL)
with the following pair of axiom schemas:
\begin{description}[labelindent=.5cm, labelwidth=1.3cm, font=\textrm]
	\item[\namedlabel{rule:mciExcMid}{\mciAxTenAxiom}] $\PropA \lor \mciNeg\PropA$
	\item[\namedlabel{rule:mcibc}{\mciBcOneAxiom}] $\mciCons\PropA \mciImp (\PropA \mciImp (\mciNeg\PropA \mciImp \PropB))$
\end{description}
%The addition of \ref{rule:mciExcMid}
%to \CPL{} produces the paraconsistent logic $PI$, studied in the 80s~\cite{batens1980}, which is not an LFI.
%An obvious way of turning it into an LFI is by adding schema \ref{rule:mcibc}, giving rise to \mbcName{}.

% non-selfextensional

% bivaluation semantics
% has a three-valued non-deterministic semantics
% possible translation semantics
% tableau

The logic \mciName{}, in turn,
is the \textbf{C}-system resulting from extending the H-system for \mbcName{}
with the following (infinitely many) axiom schemas %\cite{carniellimarcos2007,marcos2008}
\cite{marcos2008}
(the resulting \SetFmla{} H-system
is denoted here by $\mciHilbertName$):
\begin{description}[labelindent=.5cm, labelwidth=1.1cm, font=\textrm]
	\item[\namedlabel{rule:mciCi}{\mciCiAxiom}] $\mciNeg\mciCons\PropA \mciImp (\PropA \land \mciNeg\PropA)$
	%\item[\namedlabel{rule:mciCiZero}{(ci$_0$)}] $\mciCons\mciCons\FmA$
	\item[\namedlabel{rule:mciCij}{\mciCiJAxiom}] $\mciCons\mciNeg^j\mciCons\PropA$ (for all $0 \leq j < \omega$)
\end{description}
A unary connective $\ConA$ is said to constitute
a \emph{classical negation} 
in a one-dimensional logic $\SSHCR{}$ extending $\CPL{}$
in case,
for all $\FmA,\FmB \in \LangSet{\SigA}{\PropSetA}$,
$\EmptySet \SSHCR{} \FmA \lor \ConA(\FmA)$
and 
$\EmptySet \SSHCR{} \FmA \mciImp (\ConA(\FmA) \mciImp \FmB)$.
%we have $\FmA,\ConA(\FmA) \SSHCR{} \FmB$
%and $\EmptySet \SSHCR{} \FmA,\ConA(\FmA)$.
One of the main differences between \mciName{} and \mbcName{}
is that an inconsistency connective $\bullet$ may be defined in the former using the paraconsistent negation, instead of a classical negation,
by setting $\bullet\FmA \SymbDef \mciNeg\mciCons\FmA$~\cite{marcos2008}. 
%The other presentations
%of \mbcName{} presented in~\cite{carniellimarcos2007}
%were also provided to \mciName{}.
%This logic is also known to be
%expressible in a language containing
%$\bullet$ as primitive instead of $\mciCons$.

%{\color{blue}(Sequent system for mCi)}
%
%{\color{blue}(Tableaux of mCi)}
%
%{\color{blue}(Possible translation semantics, see ref. [17] in Avron's paper)}

Both logics above were presented in~\cite{carniellimarcos2007} in ways other than H-systems:
via tableau systems, via bivaluation semantics and
via possible-translations semantics.
In addition, 
% despite not being characterizable by a  finite deterministic matrix, as shown by Marcos in~\cite{marcos2008}, Avron in~\cite{avron2005modular} presented a characteristic three-valued nd-matrix for \mbcName{} and, in~\cite{avron2008}, a 5-valued non-deterministic logical matrix for \mciName{}, 
while these logics are known not to be characterizable by a single finite deterministic matrix \cite{marcos2008}, a characteristic nd-matrix is available for \mbcName{}~\cite{avron2005modular} and a 5-valued non-deterministic logical matrix is available for \mciName{}~\cite{avron2008},
witnessing the importance of non-deterministic semantics in the study %and applicability 
of non-classical logics.
Such characterizations, moreover, 
allow for the extraction of sequent-style systems for these
logics by the methodologies developed
in~\cite{avron2007,avron2005igpl}.
Since \mciName{}'s 5-valued nd-matrix will be useful for us in future sections,
we recall it below for ease of reference.

\begin{definition}
    \label{def:avrons-five-valued-matrix}
    Let
$\FiveValuesSetName \SymbDef \Set{\mcif,\mciF, \mciI, \mciT, \mcit}$
and
$\DesSetFive \SymbDef \Set{\mciI, \mciT, \mcit}$.
Define 
the $\SigMCI$-matrix
$\mciMatrixName \SymbDef \Struct{\mciAlgebraName, \DesSetFive}$
such that $\mciAlgebraName \SymbDef \Struct{\FiveValuesSetName, \AlgInterp{\cdot}{\mciAlgebraName}}$
interprets the connectives of $\SigMCI$ according to
the following:
\begin{gather*}
    \AlgInterp{\land}{\mciAlgebraName}(\ValA_1,\ValA_2) \SymbDef 
    \begin{cases}
    \Set{\mcif} & \text{if either $\ValA_1 \not\in \DesSetFive$ or $\ValA_2 \not\in \DesSetFive$}\\
    \Set{\mciI,\mcit} & \text{otherwise}
    \end{cases}\\
    \AlgInterp{\lor}{\mciAlgebraName}(\ValA_1,\ValA_2) \SymbDef 
    \begin{cases}
    \Set{\mciI,\mcit} & \text{if either $\ValA_1 \in \DesSetFive$ or $\ValA_2 \in \DesSetFive$}\\
    \Set{\mcif} & \text{if $\ValA_1$,$\ValA_2 \not\in \DesSetFive$}
    \end{cases}\\
    \AlgInterp{\mciImp}{\mciAlgebraName}(\ValA_1,\ValA_2) \SymbDef 
    \begin{cases}
    \Set{\mciI,\mcit} & \text{if either $\ValA_1 \not\in \DesSetFive$ or $\ValA_2 \in \DesSetFive$}\\
    \Set{\mcif} & \text{if $\ValA_1 \in \DesSetFive$  and $\ValA_2 \not\in \DesSetFive$}
    \end{cases}\\
    %%
    %\begin{tabular}{>{\centering\arraybackslash}m{.7cm}|>{\centering\arraybackslash}m{.7cm}|>{\centering\arraybackslash}m{.7cm}|>{\centering\arraybackslash}m{.7cm}|>{\centering\arraybackslash}m{.7cm}|>{\centering\arraybackslash}m{.7cm}}
    %      & \mciT&\mcit&\mciF&\mcif&\mciI \\
    %     \midrule
    %     $\BMatrixInterp{\mciNeg}{\MCIFiveVMatrix}$&\mciF&\mcif&\mciT&\mcit,\mciI&\mcit,\mciI\\
    %\end{tabular}\hspace{.5cm}
    \begin{tabular}{>{\centering\arraybackslash}m{.7cm}|>{\centering\arraybackslash}m{.7cm}|>{\centering\arraybackslash}m{.7cm}|>{\centering\arraybackslash}m{.7cm}|>{\centering\arraybackslash}m{.7cm}|>{\centering\arraybackslash}m{.7cm}}
          & \mcif&\mciF&\mciI&\mciT&\mcit \\
         \midrule
         $\AlgInterp{\mciNeg}{\mciAlgebraName}$&\{\mciI,\mcit\}&\{\mciT\}&\{\mciI,\mcit\}&\{\mciF\}&\{\mcif\}\\
    \end{tabular}\hspace{.5cm}    
    %%
    %\begin{tabular}{>{\centering\arraybackslash}m{.7cm}|>{\centering\arraybackslash}m{.7cm}|>{\centering\arraybackslash}m{.7cm}|>{\centering\arraybackslash}m{.7cm}|>{\centering\arraybackslash}m{.7cm}|>{\centering\arraybackslash}m{.7cm}}
    %      & \mciT&\mcit&\mciF&\mcif&\mciI \\
    %     \midrule
    %     $\BMatrixInterp{\circ}{\MCIFiveVMatrix}$&\mciT&\mciT&\mciT&\mciT&\mciF\\
    %\end{tabular}
    \begin{tabular}{>{\centering\arraybackslash}m{.7cm}|>{\centering\arraybackslash}m{.7cm}|>{\centering\arraybackslash}m{.7cm}|>{\centering\arraybackslash}m{.7cm}|>{\centering\arraybackslash}m{.7cm}|>{\centering\arraybackslash}m{.7cm}}
          & \mcif&\mciF&\mciI&\mciT&\mcit \\
         \midrule
         $\AlgInterp{\circ}{\mciAlgebraName}$&\{\mciT\}&\{\mciT\}&\{\mciF\}&\{\mciT\}&\{\mciT\}\\
    \end{tabular}
\end{gather*}
\end{definition}
%\noindent
%The intuition behind the above truth-values,
%according to~\cite{avron2008},
%is that 
%$\mciI$ is the truth-value of
%inconsistent propositions,
%$\mciT$ and $\mciF$
%are the truth-values
%of necessarily consistent
%propositions,
%and $\mcit$ and $\mcif$
%are the truth-values of
%contingently consistent propositions.
%In this way, in terms of
%preservation of designatedness,
%$\SSHCR{\mciMatrixName}$
%is the logic that
%preserves the property
%of being either inconsistent
%or (necessarily or contingently) true.

%\subsection{Nonfinitely axiomatizable in one dimension}
One might be tempted to
apply the axiomatization algorithm of~\cite{marcelino19woll}
to the finite non-deterministic logical matrix
defined above
to obtain a finite and analytic \SetSet{} system
for \mciName. However, it is not obvious, at first,
%how to obtain a discriminator for this matrix,
%that is, it is not immediate 
whether this matrix is
sufficiently expressive or not
(we will, in fact, prove that it is not).
%(it turns out that
%it is not --- see Corollary~\ref{coro:mcinotsuff}).
%
In what follows, we will show now \mciName{} is actually axiomatizable neither by a finite \SetFmla{}
H-system (first part), nor by a finite \SetSet{} H-system (second part);
it so happens, thus, that 
it was not by chance that $\mciHilbertName$ has been originally presented with infinitely many rule schemas.
%We begin by proving that
%\mciName{} is nonfinitely axiomatizable
%by a \SetFmla{} H-system.
For the first part,
we rely on the following general result:

\begin{theorem}[\cite{wojcicki1988}, Theorem 2.2.8, adapted]
\label{the:nonfinaxiomat}
Let $\SetFmlaCR{}$ be a standard Tarskian consequence relation.
Then $\SetFmlaCR{}$ is axiomatizable by a finite
\SetFmla{} H-system if, and only if,
there is no strictly increasing sequence
$\SetFmlaCR{0},\SetFmlaCR{1},\ldots,\SetFmlaCR{n},\ldots$
of standard Tarskian consequence relations such that
$\SetFmlaCR{} \;\,=\, \Supremum_{i \in \omega} \SetFmlaCR{i}$.
\end{theorem}

In order to apply the above theorem, we first present a family of finite \SetFmla{} H-systems
that, in the sequel, will be used to provide
an increasing sequence of standard Tarskian
consequence relations whose supremum is precisely \mciName{}.
Next, we show that this sequence is
stricly increasing, by employing the
matrix methodology traditionally used for showing the independence
of axioms in a proof system. 
%In this way, we will be able to conclude that
%there is no finite one-dimensional \SetFmla{} H-system for \mciName{}.

\begin{definition}
\label{def:mcik}
For each $k \in \omega$, let $\mciHilbertName^k$ be 
%the \SetFmla{} H-system
%given by the rule schemas of 
a \SetFmla{} H-system for 
positive classical
logic together with the schemas \ref{rule:mciExcMid},
\ref{rule:mcibc}, \ref{rule:mciCi}
%\ref{rule:mciCiZero}
and
\ref{rule:mciCij}, for all
$0 \leq j \leq k$.
\end{definition}

Since $\mciHilbertName^k$ may be obtained from $\mciHilbertName$
by deleting some (infinitely many) axioms, it is immediate that:

\begin{proposition}
\label{prop:mciextendsall}
	For every $k \in \omega$,
	$\SetFmlaCR{\mciHilbertName^{k}} \; \subseteq \; \SetFmlaCR{\mciName}$.
\end{proposition}

The way we define the promised increasing sequence
of consequence relations in the next result
is by taking the systems $\mciHilbertName^{k}$
with odd superscripts, namely, we will be working with the sequence
$\SetFmlaCR{\mciHilbertName^{1}}, \SetFmlaCR{\mciHilbertName^{3}},
\SetFmlaCR{\mciHilbertName^{5}}, \ldots$\;
Excluding the cases where $k$ is even will facilitate,
in particular, the proof of \autoref{lem:mcistrictinc}.

\begin{lemma}
\label{lem:mcisequence}
For each $1 \leq k < \omega$,
let 
$\SetFmlaCR{k} \;\SymbDef\; 
\SetFmlaCR{\mciHilbertName^{2k - 1}}$.
Then~$\SetFmlaCR{1} \;\subseteq\; \SetFmlaCR{2} \;\subseteq \ldots$,
and
\begin{equation*}
	\SetFmlaCR{\mciName} \;= \Supremum{}_{1 \leq k < \omega} \SetFmlaCR{k}.
\end{equation*}
\end{lemma}

Finally, we prove that the sequence outlined
in the paragraph before Lemma~\ref{lem:mcisequence} is strictly increasing.
In order to achieve this, we define, for each $1 \leq k < \omega$, a $\SigMCI$-matrix $\MatA_k$
and prove that $\mciHilbertName^{2k-1}$ is sound with respect to such matrix.
Then, in the second part of the proof (the ``independence part''), we show that, for each $1 \leq k < \omega$, 
$\MatA_k$ fails to validate
the rule schema \ref{rule:mciCij}, for $j = 2k$,
which is present in $\mciHilbertName^{2(k+1)-1}$.
In this way, by the contrapositive of the soundness
result proved in the first part, we will have~\ref{rule:mciCij} 
provable in $\mciHilbertName^{2(k+1)-1}$
while unprovable in
$\mciHilbertName^{2k-1}$.
In what follows, for any $k \in \NatSet$,
we use $\Suc{k}$ to refer to the successor of~$k$.

\begin{definition}\sloppy
	Let $1 \leq k < \omega$. 
	Define the $2\Suc{k}$-valued
	$\SigMCI$-matrix
	$\MatA_{k} \SymbDef \Struct{\AlgA_k, D_k}$
	such that
	$D_k \SymbDef \{\Suc{k} + 1, \ldots, 2\Suc{k}\}$
	and $\AlgA_k \SymbDef \Struct{\{1,\ldots,2\Suc{k}\}, \AlgInterp{\cdot}{\AlgA_k}}$, the interpretation of $\SigMCI$ in $\AlgA_k$
	given by the following operations:
	\begin{gather*}
\ValA \AlgInterp{\lor}{\AlgA_k} \ValB \SymbDef
\begin{cases}
    1 & \text{if } \ValA,\ValB \in \ValueSetComp{D_k}\\
    \Suc{k}+1 & \text{otherwise}
\end{cases}
\qquad
\ValA \AlgInterp{\land}{\AlgA_k} \ValB \SymbDef
\begin{cases}
    \Suc{k}+1 & \text{if } \ValA,\ValB \in D_k\\
    1 & \text{otherwise}
\end{cases}
\\
\ValA \AlgInterp{\mciImp}{\AlgA_k} \ValB \SymbDef
\begin{cases}
    1 & \text{if } \ValA \in D_k \text{ and } \ValB \not\in \ValueSetComp{D_k}\\
    \Suc{k}+1 & \text{otherwise}
\end{cases}\\
\AlgInterp{\mciCons}{\AlgA_k} \ValA \SymbDef
\begin{cases}
    1 & \text{if } \ValA = 2\Suc{k}\\
    \Suc{k}+1 & \text{otherwise}\\
\end{cases}
\;
\AlgInterp{\mciNeg}{\AlgA_k} \ValA \SymbDef
\begin{cases}
    \Suc{k}+1 & \text{if } \ValA \in \{1, 2\Suc{k}\}\\
    \ValA + \Suc{k} & \text{if } 2 \leq \ValA \leq \Suc{k}\\
    %i-k & \text{if } \Suc{k}+1 \leq i \leq 2\Suc{k}-1
    \ValA-(\Suc{k}-1) & \text{if } \Suc{k}+1 \leq \ValA \leq 2\Suc{k}-1
\end{cases}
\end{gather*}

\end{definition}

Before continuing, we state results concerning
this construction, which will be used in the
remainder of the current line of argumentation.
%In what follows, let $\mciNeg^0 n \SymbDef n$
%and $\mciNeg^{m+1} n \SymbDef \mciNeg(\mciNeg^m n)$
%($m \geq 0$). 
In what follows, when there is no risk of confusion, we omit the subscript `$\AlgA_k$'
from the interpretations to simplify the notation.

\begin{lemma}
\label{lem:itneg}
For all $k \geq 1$ and $1 \leq m \leq 2k$,
\[
\AlgInterp{\mciNeg}{\AlgA_k}^m(\Suc{k}+1) =
\begin{cases}
(\Suc{k}+1)+\frac{m}{2}, & \text{if } m \text{ is even}\\
1 +\frac{m+1}{2}, & otherwise\\
\end{cases}
\]
%where $\mciNeg$ is $\AlgInterp{\mciNeg}{\AlgA_k}$.
\end{lemma}
%%\begin{proof}
%%See Appendix~\ref{app:proofs}.
%Let $k \geq 1$. We prove the lemma by strong induction on $1 \leq m \leq 2k$. For $m=1$,
%we have $\mciNeg(\Suc{k}+1)=(\Suc{k}+1)-(\Suc{k}-1) = 2 = 1 + \frac{1+1}{2}$.
%	Assume now that (IH): the present lemma holds for all $m' < m$, for a given $m > 1$.
%\begin{itemize}
%    \item Suppose that $m = 2s$, with $1 \leq s \leq k$.
%	    By (IH), we have that
%    $\mciNeg^{2s}(\Suc{k}+1) = \mciNeg(\mciNeg^{2s-1}(\Suc{k}+1)) = \mciNeg(1+\frac{(2s-1)+1}{2}) = \mciNeg(1+s)$.
%    By the interpretation of $\mciNeg$, as $2 \leq 1+s \leq \Suc{k}$, we have $\mciNeg(1+s) = 1+s+\Suc{k} = (\Suc{k}+1) + \frac{m}{2}$.
%    \item Suppose that $m = 2s+1$, with $1 \leq s \leq k-1$.
%	    By (IH), we have
%    $\mciNeg^{2s+1}(\Suc{k}+1) = \mciNeg(\mciNeg^{2s}(\Suc{k}+1)) = \mciNeg(\Suc{k}+1 + \frac{2s}{2}) = \mciNeg(\Suc{k}+1+s)$. As $\Suc{k}+2 \leq \Suc{k}+1+s \leq \Suc{k}+k$, the interpretation of $\mciNeg$ gives
%    us that $\mciNeg(\Suc{k}+1+s) = (\Suc{k}+1+s) - (\Suc{k}-1) = s + 2 = \frac{m-1}{2} + 2 = (\frac{m-1}{2}+1)+1 = 1+\frac{m+1}{2}$.\qed%\qedhere
%\end{itemize}
%%\end{proof}

\begin{lemma}
\label{lem:mcistrictinc}
For all $1 \leq k < \omega$, 
we have
$\SetFmlaCR{\mciHilbertName^{2\Suc{k}-1}} \mciCons\mciNeg^{2k}\mciCons\PropA$
but 
$\not\SetFmlaCR{\mciHilbertName^{2k-1}} \mciCons\mciNeg^{2k}\mciCons\PropA$.
%where $k^\ast \SymbDef k + 1$.
\end{lemma}

Finally, Theorem~\ref{the:nonfinaxiomat},
Lemma~\ref{lem:mcisequence}
and Lemma~\ref{lem:mcistrictinc}
give us the main result:

\begin{theorem}
	\label{the:mcinotfinsetfmla}
	\mciName{} is not axiomatizable by a finite
	\SetFmla{} H-system.
\end{theorem}
%\begin{proof}
%	Direct consequence of 
%	Lemma~\ref{lem:mcisequence},
%	Lemma~\ref{lem:mcistrictinc}
%	and Theorem~\ref{the:nonfinaxiomat}.
%\end{proof}

For the second part ---namely, that no finite \SetSet{} H-system 
axiomatizes \mciName{}---, 
we make use of the following result:

\begin{theorem}[\cite{ss1978}, Theorem 5.37, adapted]
	\label{the:disjunctionsetset}
	Let $\SetSetCR{}$ be a one-dimensional consequence relation
	over a propositional signature 
	containing the binary connective $\lor$.
	Suppose that 
	the \SetFmla{} Tarskian companion of $\SetSetCR{}$,
	denoted by
	$\SetFmlaRel{\SetSetCR{}}{}$, satisfies the following
	property:
	\begin{equation}
		\label{prop:disj}
		\tag{\upshape{$\mathsf{Disj}$}}
		\FmSetA, \FmA \lor \FmB \SetFmlaRel{\SetSetCR{}}{} \FmC
		\text{ if, and only if, }
		\FmSetA, \FmA \SetFmlaRel{\SetSetCR{}}{} \FmC
		\text{ and }
		\FmSetA, \FmB \SetFmlaRel{\SetSetCR{}}{} \FmC
	\end{equation}
	If a \SetSet{} H-system $\SSHCalcA$
	axiomatizes $\SetSetCR{}$,
	then $\SSHCalcA$ may be converted
	into a \SetFmla{} H-system for
	$\SetFmlaRel{\SetSetCR{}}{}$ that
	is finite whenever $\SSHCalcA$ is finite.
\end{theorem}

It turns out that:

\begin{lemma}
	\label{lem:mcihasdisj}
	\mciName{} satisfies \upshape{(\ref{prop:disj})}.
\end{lemma}
\begin{proof}
	The non-deterministic
	semantics of \mciName{}
	gives us that,
	for all $\FmA,\FmB \in \LangSet{\SigMCI}{\PropSetA}$,
	$\FmA \SetSetCR{\mciMatrixName} \FmA\lor\FmB$;
	$\FmB \SetSetCR{\mciMatrixName} \FmA\lor\FmB$,
	and
	$\FmA \lor \FmB \SetSetCR{\mciMatrixName} \FmA, \FmB$,
	and such facts easily imply
	(\ref{prop:disj}).
\end{proof}

\begin{theorem}
	\label{the:mcinotfinsetset}
	\mciName{} is not axiomatizable by a finite
	\SetSet{} H-system.
\end{theorem}
\begin{proof}
	If $\SSHCalcA$ were
	a finite \SetSet{} H-system
	for \mciName{},
	then, by Lemma~\ref{lem:mcihasdisj}
	and Theorem~\ref{the:disjunctionsetset},
	it could be turned into a 
	finite \SetFmla{} H-system for this very logic.
	This would contradict Theorem~\ref{the:mcinotfinsetfmla}.
\end{proof}

Finding a finite one-dimensional H-system for \mciName{} (analytic or not)
over the same language, 
%in the one-dimensional environment
then,
proved to be impossible. The previous result also tells
us that there is no sufficiently expressive
non-deterministic matrix that characterizes \mciName{}
(for otherwise the recipe in~\cite{marcelino19woll}
would deliver a finite analytic \SetSet{} H-system for it), and we may conclude, in particular, that:

\begin{corollary}
	\label{coro:mcinotsuff}
	The nd-matrix $\mciMatrixName$ is
	not sufficiently expressive.
\end{corollary}

The pairs of truth-values of $\mciMatrixName$ that seem not to be separable (at least one of these pairs must not be, in view of the  above corollary) 
are $\Pair{\mcit}{\mciT}$
and $\Pair{\mcif}{\mciF}$.
The insufficiency of expressive power to take these specific pairs of values apart, however, would be circumvented if we had considered instead the matrix defined below, obtained
from $\mciMatrixName$ by changing
its set of designated values:

\begin{definition}
\label{def:mci-f-preserving}
    Let $\mciRejMatrixName \SymbDef 
    \Struct{\mciAlgebraName, \RejSetFive}$,
    where 
    $\RejSetFive \SymbDef \Set{\mcif, \mciI, \mciT}$.
\end{definition}
%
%Following the
%intuitive meaning given by
%Avron to the truth-values of
%$\mciMatrixName$,
%the logic determined by $\mciRejMatrixName$
%could be understood as preserving the truth
%of inconsistent and necessarily consistent propositions
%(as in the former matrix),
%but the \emph{falsity} (instead of the truth) of contingently consistent propositions.
Note that, in $\mciRejMatrixName$, we have $\mcit \not\in \RejSetFive$,
while $\mciT \in \RejSetFive$,
and we have that $\mcif \in \RejSetFive$, while $\mciF \not\in \RejSetFive$. Therefore, 
the single propositional variable
$\PropA$ separates in $\mciRejMatrixName$ the pairs
$\Pair{\mcit}{\mciT}$
and $\Pair{\mcif}{\mciF}$.
On the other hand,
it is not clear now whether
the pairs $\Pair{\mcit}{\mciF}$ and $\Pair{\mcif}{\mciT}$ are separable in this new matrix.
%in other words, $\mciRejMatrixName{}$ is also not sufficiently expressive.
%It turns out, however, that not all of the
%other pairs of truth-values
%are separable in $\mciRejMatrixName$,
%as they were in $\mciMatrixName$:
%\begin{proposition}
%The pairs $\Pair{\mcit}{\mciF}$ and $\Pair{\mcif}{\mciT}$ are not separable in
%$\mciRejMatrixName$.
%Thus, $\mciRejMatrixName$ is also not sufficiently expressive.
%\end{proposition}
%
%Even though $\mciMatrixName$ and $\mciRejMatrixName$
%are not sufficiently expressive on their own,
%Nevertheless, it seems that the expressiveness power
%of $\mciMatrixName{}$ and $\mciRejMatrixName{}$, 
%if combined somehow,
%would be enough to separate every truth-value of
%$\mciAlgebraName$, using the same language. 
%In the next section, 
Nonetheless, we will see, in the next section,
how we can take advantage
of the semantics of non-deterministic \TheB-matrices in order to 
combine the expressiveness of 
$\mciMatrixName$ and $\mciRejMatrixName$ 
in a very simple and intuitive manner, preserving
the language and the algebra shared by these matrices.
The notion of logic induced by the resulting structure will not
be one-dimensional, as the one presented before, but
rather two-dimensional, in a sense we shall detail
in a moment. We identify two important aspects
of this combination: first, the logics determined
by the original matrices can be fully recovered from the combined logic;
and, second, since the notions of H-systems and sufficient expressiveness,
as well as the axiomatization algorithm of~\cite{marcelino19woll},
were generalized in~\cite{greati2021}, the resulting two-dimensional logic
may be algorithmically axiomatized by an \emph{analytic} two-dimensional
H-system that is \emph{finite} if the combining matrices are finite,
provided the criterion of sufficient expressiveness is satisfied
after the combination. 
This will be the case, in particular, when we combine
$\mciMatrixName$ and $\mciRejMatrixName$.
Consequently, this novel way of
combining logics provides a quite general approach for producing finite and analytic
axiomatizations for logics determined by non-deterministic logical matrices 
that fail to be 
finitely axiomatizable in one dimension;
this includes the logics from Example~\ref{ex:three-valued-non-ax-set-set},
and also \mciName{}.
%whose induced two-dimensional consequence relation
%has in its $\tAsp$-aspect the logic \mciName{}.

\section{Two-dimensional logics}%: more expressiveness, same object-language}
\vspace{-.5em}
\label{sec:two-dim}
From now on, we will employ the symbols
$\Acc$, $\NAcc$, $\Rej$ and $\NRej$ to informally refer to, respectively, the cognitive attitudes of
%the cognitive attitudes~\cite{jm} of
\emph{acceptance}, \emph{non-acceptance},
\emph{rejection} and \emph{non-rejection}, 
collected in the set $\CogsAttSet \SymbDef \Set{\Acc,\NAcc,\Rej,\NRej}$.
Given a set $\FmSetA \subseteq \LangSet{\SigA}{\PropSetA}$,
we will write $\FmSetA_{\CogVarA}$ to intuitively mean
that a given agent entertains the cognitive attitude
$\CogVarA \in \CogsAttSet$ with respect to the formulas in $\FmSetA$, 
that is: the formulas in $\FmSetA_\Acc$
will be understood as being accepted by the agent; 
the ones in $\FmSetA_\NAcc$, as non-accepted;
the ones in $\FmSetA_\Rej$, as rejected; and
the ones in $\FmSetA_\NRej$, as non-rejected.
Where $\CogVarA \in \CogsAttSet$, we let
$\InvCog{\CogVarA}$ be its flipped version,
that is,
$\InvCog{\Acc} \SymbDef \NAcc$,
$\InvCog{\NAcc} \SymbDef \Acc$,
$\InvCog{\Rej} \SymbDef \NRej$ and
$\InvCog{\NRej} \SymbDef \Rej$.

We refer to each
$\BStat[]{\CtxAccA}{\CtxNRejA}{\CtxRejA}{\CtxNAccA} \in \PowerSet{\LangSet{\SigA}{\PropSetA}}^2 \times \PowerSet{\LangSet{\SigA}{\PropSetA}}^2$ as a \emph{\TheB-statement},
where $\Pair{\CtxAccA}{\CtxRejA}$ is the \DefEmph{antecedent} and
$\Pair{\CtxNAccA}{\CtxNRejA}$ is the \DefEmph{succedent}.
The sets in the latter pairs are called \DefEmph{components}.
A \TheB\DefEmph{-consequence relation} 
is a collection $\BConName{}$ of \TheB-statements satisfying:
%$\TwoXTwo$-place relation
%    $\BConName \subseteq \PowerSet{\LangSet{\SigA}{\PropSetA}}^2 \times \PowerSet{\LangSet{\SigA}{\PropSetA}}^2$ satisfying:
%any of the following conditions constitute sufficient guarantee
%for the \emph{\TheB-consequence judgement}
%$\BCon[]{\CtxAccA}{\CtxNRejA}{\CtxRejA}{\CtxNAccA}$
%to be established:

\begin{description}[labelindent=0.5cm, labelwidth=1cm]\itemsep1pt
	\item[\namedlabel{prop:BConO}\BConO] 
		%$\BCon[d]{\CtxAccA}{\CtxNRejA}{\CtxRejA}{\CtxNAccA}$,
		%whenever 
		if	$\CtxAccA \cap \CtxNAccA \neq \EmptySet$
		or 
			$\CtxRejA \cap \CtxNRejA \neq \EmptySet$,
			then $\BCon[]{\CtxAccA}{\CtxNRejA}{\CtxRejA}{\CtxNAccA}$
	\item[\namedlabel{prop:BConD}\BConD]
		%$\BCon[d]{\CtxAccB}{\CtxNRejB}{\CtxRejB}{\CtxNAccB}$,
		%whenever 
		if $\BCon[]{\CtxAccB}{\CtxNRejB}{\CtxRejB}{\CtxNAccB}$ and 	$\CogSet{\FmSetB}{\CogVarA} \subseteq \CogSet{\FmSetA}{\CogVarA}$
		for every $\CogVarA \in \CogsAttSet$,
		then
		$\BCon[]{\CtxAccA}{\CtxNRejA}{\CtxRejA}{\CtxNAccA}$
	\item[\namedlabel{prop:BConC}\BConC] 
		%$\BCon[d]{\CtxAccA}{\CtxNRejA}{\CtxRejA}{\CtxNAccA}$,
		%whenever
		if
			$\BCon[]{\IntermAccSet{\CutPropSetA}}{\FmSetCompl{\IntermRejSet{\CutPropSetA}}}{\IntermRejSet{\CutPropSetA}}{\FmSetCompl{\IntermAccSet{\CutPropSetA}}}$ %\\
		for all 
			$\CtxAccA \subseteq \IntermAccSet{\CutPropSetA} \subseteq \FmSetCompl{\CtxNAccA}$
		and
			$\CtxRejA \subseteq \IntermRejSet{\CutPropSetA} \subseteq \FmSetCompl{\CtxNRejA}$,
			then $\BCon[]{\CtxAccA}{\CtxNRejA}{\CtxRejA}{\CtxNAccA}$
\end{description}

\noindent A \TheB-consequence relation is called \DefEmph{substitution-invariant} if, in addition, 
    $\BCon[]{\CtxAccA}{\CtxNRejA}{\CtxRejA}{\CtxNAccA}$ holds
whenever, for every $\SubstA\in\SubsSet{\SigA}{}$:
\begin{description}[labelindent=0.5cm, labelwidth=1cm]\itemsep1pt
	\item[\namedlabel{prop:BConS}\BConS]
		$\BCon[]{\CtxAccB}{\CtxNRejB}{\CtxRejB}{\CtxNAccB}$
		and
		$\CogSet{\FmSetA}{\CogVarA} = \SubsApply{\SubstA}{\CogSet{\FmSetB}{\CogVarA}}$
		for every $\CogVarA \in \CogsAttSet$ 
		%and some
		%%substitution~
		%$\SubstA\in\SubsSet{\SigA}{}$
		%
		%$\BCon[d]
		%{\SubsApply{\SubstA}{\CtxAccA}}
		%{\SubsApply{\SubstA}{\CtxNRejA}}
		%{\SubsApply{\SubstA}{\CtxRejA}}
		%{\SubsApply{\SubstA}{\CtxNAccA}}$
		%for every 
		%	$\SubstA \in \EndSet{\DefLangAlg}$,
		%whenever 
		%	$\BCon[d]{\CtxAccA}{\CtxNRejA}{\CtxRejA}{\CtxNAccA}$
\end{description}

\noindent Moreover, a \TheB-consequence relation is called \DefEmph{finitary}
when it enjoys the property

\begin{description}[labelindent=0.5cm, labelwidth=1cm]\itemsep1pt
    \item[\namedlabel{prop:BConF}\BConF]
    if
    $\BCon[]{\CtxAccA}{\CtxNRejA}{\CtxRejA}{\CtxNAccA}$,
    then
    	$\BCon[]
    	{\FinSubset{\CtxAccA}}
    	{\FinSubset{\CtxNRejA}}
    	{\FinSubset{\CtxRejA}}
    	{\FinSubset{\CtxNAccA}}$, %\\
	for some finite 
		$\FinSubset{\CogSet{\FmSetA}{\CogVarA}} \subseteq \CogSet{\FmSetA}{\CogVarA}$,
	and each $\CogVarA \in \CogsAttSet$    
\end{description}
In what follows, \TheB-consequence
relations will also be referred
to as \DefEmph{two-dimen\-sion\-al logics}.
The complement of $\BConName$, sometimes called the \emph{compatibility relation associated with} $\BConName$~\cite{blasio2021}, will be denoted by~$\nBConName$.
%We will denote by $\nBConName$ the complement of $\BConName$, sometimes called the \emph{compatibility relation associated with} $\BConName$~\cite{blasio2021}.
%
Every \TheB-consequence relation $\BConNameLetter \SymbDef \BConName$
induces one-dimensional consequence relations
$\SetSetRel{\tAsp}{\BConNameLetter}$
and $\SetSetRel{\fAsp}{\BConNameLetter}$,
such that
$\CtxAccA \SetSetRel{\tAsp}{\BConNameLetter} 
\CtxNAccA$ % if, and only if,
iff
$\BCon[]{\CtxAccA}{\EmptySet}{\EmptySet}{\CtxNAccA}$,
and
$\CtxRejA \SetSetRel{\fAsp}{\BConNameLetter} 
\CtxNRejA$ %if, and only if,
iff
$\BCon[]{\EmptySet}{\CtxNRejA}{\CtxRejA}{\EmptySet}$.
Given a one-dimensional consequence relation $\SetSetCR{}$,
we say that it
\emph{inhabits the $\tAsp$-aspect of $\BConNameLetter$}
if $\SetSetCR{} = \SetSetRel{\tAsp}{\BConNameLetter}$,
and that it
\emph{inhabits the $\fAsp$-aspect of $\BConNameLetter$}
if $\SetSetCR{} = \SetSetRel{\fAsp}{\BConNameLetter}$.
\TheB-consequence relations
actually induce many other (even non-Tarskian) 
one-dimensional notions of logics;
the reader is referred to~\cite{blasio20171,blasiomarcos2017} for a
thorough presentation on this topic.

As we did for one-dimensional consequence relations,
we present now realizations of \TheB-consequence relations, first via the semantics of nd-\TheB-matrices, then
by means of two-dimensional H-systems.

%\subsection{Non-deterministic matrix structures}

A \DefEmph{non-deterministic \TheB-matrix over $\SigA$},
or simply \DefEmph{$\SigA$-nd-\TheB-matrix},
is a structure $\BMatA \SymbDef \Struct{\AlgA, \BMatDesSetA{}, \BMatAntiDesSetA{}}$, where $\AlgA$ is a $\SigA$-nd-algebra,
$\BMatDesSetA{} \subseteq \ValSetA$ is the
set of \DefEmph{designated values} 
and $\BMatAntiDesSetA{} \subseteq \ValSetA$
is the set of \DefEmph{antidesignated values}
of $\BMatA$. For convenience, we define
$\BMatNDesSetA{} \SymbDef \SetDiff{\ValSetA}{\BMatDesSetA{}}$ to be
the set of \DefEmph{non-designated values},
and $\BMatNAntiDesSetA{} \SymbDef \SetDiff{\ValSetA}{\BMatAntiDesSetA{}}$ to be the set of \DefEmph{non-antidesignated values}
of $\BMatA$.
%Given $X \subseteq \ValSetA$, the
%\DefEmph{sub-$\SigA$-nd-\TheB-matrix
%induced by $X$} is given by
%$\IndSubBMatrix{\BMatA}{X} \SymbDef 
%\Struct{\IndSubalgebra{\AlgA}{X}, \BMatDesSetA{} \cap X, \BMatAntiDesSetA{} 
%\cap X}$.
The elements of $\ValSet{\SigA}{\AlgA}$ are dubbed
\DefEmph{$\BMatA$-valuations}.
%Let $\BMatA \SymbDef \Struct{\AlgA, \BMatDesSetA{}, \BMatAntiDesSetA{}}$ be a $\SigA$-nd-\TheB-matrix.
The
\DefEmph{\TheB-entailment relation determined by $\BMatA$}
 is a collection $\BEntName{\BMatA}$ of \TheB-statements
such that
%is a $\TwoXTwo$-place relation $\BEntName{\BMatA}$ over $\LangSet{\SigA}{\PropSetA}$ such that
\begin{description}[labelindent=.5cm, labelwidth=1cm]\itemsep1pt    %\begin{namedproperties}
\item[\namedlabel{prop:BentDef}\BEntProp]
    \begin{tabular}{lcl}
    $\BEnt[d]{\CtxAccA}{\CtxNRejA}{\CtxRejA}
    {\CtxNAccA}{\BMatA}$
        & \ \ if{f}\ \ &
        \begin{tabular}{l}
    there is no $\BMatA$-valuation
    $\ValuationA$ such that
\\
        $\ValuationA(\CogSet{\FmSetA}{\CogVarA})
        \subseteq \BMatDistSet{\CogVarA}{\BMatA}$
        for each $\CogVarA \in \CogsAttSet$,
        \end{tabular}
    \end{tabular}
%\end{namedproperties}
\end{description}
for every $\CtxAccA, \CtxRejA, \CtxNAccA, \CtxNRejA \subseteq \LangSet{\SigA}{\PropSetA}$.
Whenever $\BEnt[]{\CtxAccA}{\CtxNRejA}{\CtxRejA}{\CtxNAccA}{\BMatA}$, we say that
the \TheB-statement $\BStat[]{\CtxAccA}{\CtxNRejA}{\CtxRejA}{\CtxNAccA}$
\emph{holds in $\BMatA$} or \emph{is valid in $\BMatA$}.
An $\BMatA$-valuation that bears witness to
$\nBEnt[]{\CtxAccA}{\CtxNRejA}{\CtxRejA}{\CtxNAccA}{\BMatA}$
is called a \DefEmph{countermodel
for $\BStat[]{\CtxAccA}{\CtxNRejA}{\CtxRejA}{\CtxNAccA}$
in $\BMatA$}.
%As proved in~\ref{vitormsc},
One may easily check that
$\BEntName{\BMatA}$ is a substitution-invariant \TheB-con\-se\-quen\-ce relation,
that is finitary 
when $\ValSetA$ is finite.
Taking $\BConNameLetter$ as  $\BEntName{\BMatA}$,
we define
$\SetSetRel{\tAsp}{\BMatA} \SymbDef \SetSetRel{\tAsp}{\BConNameLetter}$
and
$\SetSetRel{\fAsp}{\BMatA} \SymbDef \SetSetRel{\fAsp}{\BConNameLetter}$.

%%%%%%%
We move now to two-dimensional, or \SetTSetT{}, H-systems,
first introduced in~\cite{greati2021}.
A \DefEmph{(schematic) \SetTSetT{} rule of inference} $\BRuleOfInfA_\BRuleSchemaA$ 
is the collection
of all substitution instances
of the \SetTSetT{} statement $\BRuleSchemaA$,
called the \emph{schema} of $\BRuleOfInfA_\BRuleSchemaA$.
%$\Pair{\FmSetC}{\FmSetD} \in \PowerSet{\LangSet{\SigA}{\PropSetA}}\times\PowerSet{\LangSet{\SigA}{\PropSetA}}$, called rule instances and usually denoted by $\HRule{\FmSetC}{\FmSetD}$,
Each 
%$\BRuleInst \SymbDef \BStat[]{\CtxAccA}{\CtxNRejA}{\CtxRejA}{\CtxNAccA}{} \in  \BRuleOfInfA_\BRuleSchemaA$
$\BRuleInst \in  \BRuleOfInfA_\BRuleSchemaA$
is said to be a \emph{rule instance of $\BRuleOfInfA_\BRuleSchemaA$}.
%% We denote schemas and rule instances by $\BRuleIn[]{\CtxAccA}{\CtxNRejA}{\CtxRejA}{\CtxNAccA}$ rather than by $\BStat[]{\CtxAccA}{\CtxNRejA}{\CtxRejA}{\CtxNAccA}{}$.
In a proof-theoretic context, rather than writing the \TheB-statement $\BStat[]{\CtxAccA}{\CtxNRejA}{\CtxRejA}{\CtxNAccA}{}$, we shall denote the corresponding rule by $\BRuleIn[]{\CtxAccA}{\CtxNRejA}{\CtxRejA}{\CtxNAccA}$. 
%
%where $\Pair{\CtxAccA}{\CtxRejA}$ is the \DefEmph{antecedent} and
%$\Pair{\CtxNAccA}{\CtxNRejA}$ is the \DefEmph{succedent}.
A \DefEmph{(schematic) \SetTSetT{} H-system} $\BCalcA$ is a collection of
\SetTSetT{} rules of inference.
%%%
	\DefEmph{\SetTSetT{} derivations} are as in
	the \SetSet{} H-systems,
	but now 
	the nodes are labelled with pairs of sets of
	formulas, instead of a single set.
	When applying a rule instance, 
	each formula in the succedent produces a new branch
	as before, but now the formula goes to
	the same component in which it was
	found in the rule instance.
	See Figure~\ref{fig:bderivationscheme}
	for a general representation
	and compare it with Figure~\ref{fig:derivationscheme}.
	
%%%%%%%%%%
%\vspace{-1.5em}
%%%%%%%%%%

	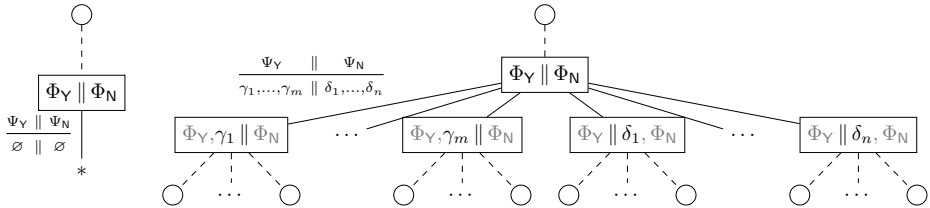
\begin{figure}[htb]
    \centering
    \scalebox{0.8}{
    \begin{tikzpicture}[every tree node/.style={},
       level distance=1.3cm,sibling distance=1cm,
       edge from parent path={(\tikzparentnode) -- (\tikzchildnode)}, baseline]
        \Tree[.\node[style={draw,circle}] {};
            \edge[dashed];
            [.\node[style={draw}] {$\BNode{\CtxAccA}{\CtxRejA}$};
            \edge node[auto=right] {$\BRuleIn{\CtxAccB}{\EmptySet}{\CtxRejB}{\EmptySet}$};
            [.{$\StarLabel$}
            ]
            ] 
        ] 
    \end{tikzpicture}
    }\quad
    \scalebox{0.8}{
    \begin{tikzpicture}[every tree node/.style={},
       level distance=1cm,sibling distance=.45cm,
       edge from parent path={(\tikzparentnode) -- (\tikzchildnode)}, baseline]
        \Tree[.\node[style={draw,circle}] {};
            \edge[dashed];
            [.\node[style={draw}] {$\BNode{\CtxAccA}{\CtxRejA}$};
            \edge node[auto=right] {$\BRuleIn{\CtxAccB}{\FmD_1,\ldots,\FmD_n}{\CtxRejB}{\FmC_1,\ldots,\FmC_m}$};
            [.\node[style={draw}] {$\BNode{\mkgray{\CtxAccA,} \FmC_1}{\mkgray{\CtxRejA}}$};
                \edge[dashed];
                [.\node[style={draw,circle}]{};
                ]
                \edge[dashed];
                [.\node[style={}]{$\cdots$};
                ]
                \edge[dashed];
                [.\node[style={draw,circle}]{};
                ]
            ]
            [.$\ldots$
            ]
            [.\node[style={draw}] {$\BNode{\mkgray{\CtxAccA,}\FmC_m}{\mkgray{\CtxRejA}}$};
                \edge[dashed];
                [.\node[style={draw,circle}]{};
                ]
                \edge[dashed];
                [.\node[style={}]{$\cdots$};
                ]
                \edge[dashed];
                [.\node[style={draw,circle}]{};
                ]
            ]
            [.\node[style={draw}] {$\BNode{\mkgray{\CtxAccA}}{\FmD_1\mkgray{,\CtxRejA}}$};
                \edge[dashed];
                [.\node[style={draw,circle}]{};
                ]
                \edge[dashed];
                [.\node[style={}]{$\cdots$};
                ]
                \edge[dashed];
                [.\node[style={draw,circle}]{};
                ]
            ]
            [.$\ldots$
            ]
            [.\node[style={draw}] {$\BNode{\mkgray{\CtxAccA}}{\FmD_n\mkgray{,\CtxRejA}}$};
                \edge[dashed];
                [.\node[style={draw,circle}]{};
                ]
                \edge[dashed];
                [.\node[style={}]{$\cdots$};
                ]
                \edge[dashed];
                [.\node[style={draw,circle}]{};
                ]
            ]
            ]
        ]
    \end{tikzpicture}
    }
    \caption{Graphical representation of finite $\BCalcA$-derivations.
    %, in the form of a leaf node, a discontinued node and an expanded node, respectively.
    We emphasize that, in both cases,
    we must have
    $\CtxAccB \subseteq \CtxAccA$ and $\CtxRejB \subseteq \CtxRejA$
    to enable the application of the rule.}
    \label{fig:bderivationscheme}
\end{figure}

%%%%%%%%%%
%\vspace{-1.5em}
%%%%%%%%%%

Let $\BTreeA$ be an $\BCalcA$-derivation.
A node $\BNodeA$ of $\BTreeA$ is
\DefEmph{$\Pair{\CtxNAccB}{\CtxNRejB}$-closed}
in case it is discontinued (namely, labelled with $\StarLabel$) or
it is a leaf node with
$\LabelFn{\BTreeA}(\BNodeA) = \Pair{\CtxAccA}{\CtxRejA}$
and either $\CtxAccA \cap \CtxNAccB \neq \EmptySet$
or $\CtxRejA \cap \CtxNRejB \neq \EmptySet$.
A branch of $\BTreeA$ is 
\DefEmph{$\Pair{\CtxNAccB}{\CtxNRejB}$-closed}
when it ends in a
$\Pair{\CtxNAccB}{\CtxNRejB}$-closed
node.
An $\BCalcA$-derivation $\BTreeA$
is said to be
\DefEmph{$\Pair{\CtxNAccB}{\CtxNRejB}$-closed}
when all of its branches are
$\Pair{\CtxNAccB}{\CtxNRejB}$-closed.
An \DefEmph{$\BCalcA$-proof} of
$\BStat{\CtxAccA}{\CtxNRejA}{\CtxRejA}{\CtxNAccA}$
is a $\Pair{\CtxNAccA}{\CtxNRejA}$-closed
$\BCalcA$-derivation $\BTreeA$ with
$\LabelFn{\BTreeA}(\TreeRoot{\BTreeA}) \SubseteqPair \Pair{\CtxAccA}{\CtxRejA}$.
The definitions of the 
(finitary) substitution-invariant \TheB-consequence relation $\BConCalcName{\BCalcA}$ induced by a (finitary) \SetTSetT{} H-system
$\BCalcA$ and $\BAnaSetA$-analyticity
are obvious generalizations of the corresponding \SetSet{} definitions.

In~\cite{greati2021}, the notion of
sufficient expressiveness was generalized
to nd-\TheB-matrices. We reproduce here
the main definitions for self-containment:
\vspace{-2mm}

\begin{definition}
    Let $\BMatA \SymbDef \Struct{\AlgA, \BMatDesSetA{\BMatA}, \BMatAntiDesSetA{\BMatA}}$
    be a $\SigA$-nd-\TheB-matrix. 
    \vspace{-2mm}
    \begin{itemize}
        \item Given
                $\ValSubsetA,\ValSubsetB \subseteq \ValSetA$
                and $\CogVarA \in \{\Acc,\Rej\}$,
                we say that $\ValSubsetA$
                and $\ValSubsetB$ are \DefEmph{$\CogVarA$-separated},
                denoted by $\Separated{\ValSubsetA}{\ValSubsetB}{\CogVarA}$,
                if $\ValSubsetA \subseteq \CogVarA$
                and $\ValSubsetB \subseteq \InvCog{\CogVarA}$,
                or vice-versa.
        \item Given distinct truth-values $\ValA, \ValB \in \ValSetA$,
		a unary formula $\SepA$ is a \DefEmph{separator for $\Pair{\ValA}{\ValB}$} whenever
            $\Separated{\AlgInterp{\SepA}{\AlgA}(\ValA)}{\AlgInterp{\SepA}{\AlgA}(\ValB)}{\CogVarA}$
            for some $\CogVarA \in \{\Acc,\Rej\}$.
            If 
            there is a separator 
            for each pair of distinct truth-values in $\ValSetA$, 
            %there is a separator for these values, 
            then $\BMatA$ is said to be
            \DefEmph{sufficiently expressive}.
%        \item A set of unary formulas 
%        $\DiscriminatorVal{\ValA}{}{\BMatA}$
%        \DefEmph{isolates} $\ValA \in \ValSetA$
%        whenever, for every $\ValB \neq \ValA$, there exists
%        a separator in
%        $\DiscriminatorVal{\ValA}{}{\BMatA}$
%        for $\ValA$ and $\ValB$.
%        \item A \DefEmph{discriminator for $\BMatA$}
%        is a family $\DiscriminatorA \SymbDef 
%        \{\Tuple{
%            \DiscriminatorVal{\ValA}{\Acc}{\BMatA},
%            \DiscriminatorVal{\ValA}{\NAcc}{\BMatA},
%            \DiscriminatorVal{\ValA}{\Rej}{\BMatA},
%            \DiscriminatorVal{\ValA}{\NRej}{\BMatA}
%        }\}_{\ValA \in \ValSetA}$
%        such that $\DiscriminatorVal{\ValA}{}{\BMatA} 
%        \SymbDef \bigcup_{\CogVarA \in \CogsAttSet} \DiscriminatorVal{\ValA}{\CogVarA}{\BMatA}$
%        isolates $\ValA$ and 
%        $\AlgInterp{\SepA}{\AlgA}(\ValA) \subseteq \CogVarA$
%        whenever $\SepA \in  \DiscriminatorVal{\ValA}{\CogVarA}{\BMatA}$.
%        We denote the set $\bigcup_{\ValA \in \ValSetA}\bigcup_{\CogVarA \in \CogsAttSet} \DiscriminatorVal{\ValA}{\CogVarA}{\BMatA}$ 
%        by $\DiscSet{\DiscriminatorA}$
%        and say that $\DiscriminatorA$ is \DefEmph{based on $\DiscSet{\DiscriminatorA}$}.
    \end{itemize}
\end{definition}
\vspace{-2mm}

\noindent In the same work~\cite{greati2021}, the axiomatization algorithm
of~\cite{marcelino19woll} was also generalized,
guaranteeing that
every sufficiently expressive nd-\TheB-matrix $\BMatA$
is axiomatizable by a $\BAnaSetA$-analytic \SetTSetT{} H-system,
which is finite whenever~$\BMatA$ is finite,
where~$\BAnaSetA$ is a set of separators
for the pairs of truth-values of $\BMatA$.
Note that, in the second bullet of the above
definition, a unary formula is characterized as
a separator whenever it separates a pair of truth-values
according to \emph{at least one} of the distinguished sets of values.
This means that having two of such sets may allow us
to separate more pairs of truth-values than having a single set, 
that is, 
the nd-\TheB-matrices are, in this sense, potentially more
expressive than the (one-dimensional) logical matrices.

\begin{example}
\label{ex:2d-axiom-non-fin}
Let $\AlgA$ be the $\SigA$-nd-algebra 
from Example~\ref{ex:three-valued-non-ax-set-set},
and consider the nd-\TheB-matrix
$\BMatA \SymbDef \Struct{\AlgA,\Set{\tVal},\Set{\fVal}}$.
As we know, in this matrix the pair
$\Pair{\fVal}{\BotVal}$ is not separable if
we consider only the set of designated values
$\Set{\tVal}$. However, as we have now
the set $\Set{\fVal}$ of antidesignated truth-values,
the separation becomes evident: the
propositional variable $\PropA$ is a separator
for this pair now, since
$\fVal \in \Set{\fVal}$
and $\BotVal \not\in \Set{\fVal}$.
The recipe from~\cite{greati2021} 
produces the following \SetTSetT{} axiomatization
for $\BMatA$, with only three very simple schematic rules
of inference:
\vspace{-2mm}
\begin{gather*}
    \BRuleIn[d]{\PropA}{}{\PropA}{}\qquad
    \BRuleIn[d]{}{\PropA}{}{f(\PropA),\PropA}\qquad
    \BRuleIn[d]{}{t(\PropA)}{\PropA}{}
\end{gather*}
By construction, the one-dimensional logic
determined by the nd-matrix of
Example~\ref{ex:three-valued-non-ax-set-set}
inhabits the $\tAsp$-aspect of
$\BEntName{\BMatA}$, thus 
it can be seen as being axiomatized by
this \DefEmph{finite} and \emph{analytic} two-dimensional system
(contrast with the \emph{infinite} \SetSet{} axiomatization
known for this logic provided in that same example).
\end{example}

We constructed above a $\SigA$-nd-\TheB-matrix 
from two $\SigA$-nd-matrices in such a way that
the one-dimensional logics 
determined by latter are fully recoverable from the former.
We formalize this construction below:

\begin{definition}
    Let $\MatA \SymbDef \Struct{\AlgA,\DesSetA}$
    and $\MatA' \SymbDef \Struct{\AlgA,\DesSetA'}$
    be $\SigA$-nd-matrices.
    The \DefEmph{\TheB-product}
    between $\MatA$ and $\MatA'$ is the $\SigA$-nd-\TheB-matrix
    $\MatA \MatBProdFn \MatA' \SymbDef \Struct{\AlgA,\DesSetA,\DesSetA'}$.
\end{definition}

\noindent Note that $\FmSetA\SetSetCR{\MatA}\FmSetB$ if{f}
$\BEnt[]{\FmSetA}{}{}{\FmSetB}{\MatA \MatBProdFn \MatA'}$
if{f} $\FmSetA\SetSetRel{\tAsp}{\MatA \MatBProdFn \MatA'}\FmSetB$,
and $\FmSetA\SetSetCR{\MatA'}\FmSetB$ if{f}
$\BEnt[]{}{\FmSetB}{\FmSetA}{}{\MatA \MatBProdFn \MatA'}$
if{f} $\FmSetA\SetSetRel{\fAsp}{\MatA \MatBProdFn \MatA'}\FmSetB$.
Therefore, $\SetSetCR{\MatA}$ and $\SetSetCR{\MatA'}$
are easily recoverable 
from $\BEntName{\MatA \MatBProdFn \MatA'}$,
since they inhabit, respectively,
the $\tAsp$-aspect and the $\fAsp$-aspect of the latter.
One of the applications of this novel way of putting two distinct logics
together was illustrated in that same Example~\ref{ex:2d-axiom-non-fin} to produce a two-dimensional
analytic and finite axiomatization for
a one-dimensional logic characterized by a $\SigA$-nd-matrix.
As we have shown, the latter one-dimensional logic does not need
to be finitely axiomatizable by a \SetSet{} H-system.
We present this application of \TheB-products with more generality
below:

\begin{proposition}
\label{prop:combine-matrices}
Let $\MatA \SymbDef \Struct{\AlgA,\DesSetA}$
be a $\SigA$-nd-matrix and suppose that $U \subseteq \ValSetA\times\ValSetA$ contains
all and only the pairs
of distinct truth-values
that fail to be
separable in $\MatA$.
If, for some $\MatA' \SymbDef \Struct{\AlgA,\DesSetA'}$,
the pairs in $U$ are separable in $\MatA'$, then
 $\MatA \MatBProdFn \MatA'$
is sufficiently expressive
(thus, axiomatizable by an analytic \SetTSetT{} H-system, that is finite whenever $\AlgA$ is finite).
%Let $\MatA \SymbDef \Struct{\AlgA,\DesSetA}$
%be a $\SigA$-nd-matrix and suppose that
%\RedAlert{only} the pairs $\Pair{\ValA_1}{\ValB_1}, \ldots, \Pair{\ValA_n}{\ValB_n}, \ldots \in \ValSetA\times\ValSetA$ 
%of distinct truth-values are not
%separable in $\MatA$.
%If, for some $\MatA' \SymbDef \Struct{\AlgA,\DesSetA'}$,
%these pairs are separable in $\MatA'$, then
% $\MatA \MatBProdFn \MatA'$
%is sufficiently expressive
%(thus, axiomatizable by an analytic \SetTSetT{} H-system, that is finite whenever $\AlgA$ is finite).
\end{proposition}
%%\begin{proof}
%%See Appendix~\ref{app:proofs}.
%%\end{proof}

%In the sequel, we will apply the above result in the case of \mciName{}, providing, 
%%up to our knowledge, 
%the first finite and analytic Hilbert-style axiomatization of this LFI.

%\subsection{
%\section{A finite and analytic two-dimensional H-system for \mciName{}}
\section{A finite and analytic proof system for \mciName{}}
\label{sec:mci-two-d}
In the spirit of~\autoref{prop:combine-matrices},
we define below a nd-\TheB-matrix
by combining 
the matrices 
$\mciMatrixName \SymbDef \Struct{\mciAlgebraName, \DesSetFive}$ and 
$\mciRejMatrixName \SymbDef \Struct{\mciAlgebraName, \RejSetFive}$
introduced in Section~\ref{sec:mci-non-finit} (\autoref{def:avrons-five-valued-matrix}
and \autoref{def:mci-f-preserving}):

\begin{definition}
    \label{def:mci-b-matrix}
    Let
    $\mciBMatrixName \SymbDef \mciMatrixName \MatBProdFn \mciRejMatrixName = \Struct{\mciAlgebraName, \DesSetFive, \RejSetFive}$,
    with
    %such that
    %$\mciAlgebraName$
    %is as in
    %\autoref{def:avrons-five-valued-matrix},
    $\DesSetFive \SymbDef \Set{\mciI, \mciT, \mcit}$
    and
    $\RejSetFive \SymbDef \Set{\mcif,\mciI,\mciT}$.
    %$\RejSetFive \SymbDef \Set{\mciT, \mciI, \mcif}$.
\end{definition}

When we consider now both sets $\DesSetFive$
and $\RejSetFive$ of
designated and 
antidesignated truth-values,
the separation
of all truth-values of $\mciAlgebraName$ becomes possible,
that is, $\mciBMatrixName$ is sufficiently expressive,
as guaranteed by~\autoref{prop:combine-matrices}.
Furthermore, notice that
we have two alternatives for separating
the pairs $\Pair{\mciI}{\mcit}$
and $\Pair{\mciI}{\mciT}$:
either using the formula
$\mciNeg\PropA$
or the formula $\mciCons\PropA$.
%In Example~\ref{ex:discfourmci},
%we presented a discriminator for $\MCIFiveVMatrix$
%in which $\mciNeg\PropA$ was chosen
%as separator for the mentioned cases. 
%Below we provide a discriminator
%using $\mciCons\PropA$ to separate such pairs:
%
%\begin{table}[H]
%	\centering
%        \begin{tabular}{>{\centering\arraybackslash}m{.7cm}!{\vrule width .6pt}>{\centering\arraybackslash}m{.7cm}>{\centering\arraybackslash}m{.7cm}>{\centering\arraybackslash}m{.7cm}>{\centering\arraybackslash}m{.7cm}}
%            \toprule
%             $\ValA$ &$\DiscriminatorVal{\ValA}{\Acc}{\BMatA}$& $\DiscriminatorVal{\ValA}{\NAcc}{\BMatA}$&
%             $\DiscriminatorVal{\ValA}{\Rej}{\BMatA}$&
%             $\DiscriminatorVal{\ValA}{\NRej}{\BMatA}$\\
%             \midrule
%             \mcif&$\EmptySet$&$\PropA$&$\PropA$&$\EmptySet$\\
%             \mciF&$\EmptySet$&$\PropA$&$\EmptySet$&$\PropA$\\
%             \mciI&$\PropA$&$\mciCons\PropA$ &$\PropA$&$\EmptySet$\\
%             \mciT&$\PropA,\mciCons\PropA$& $\EmptySet$ &$\PropA$&$\EmptySet$\\
%             \mcit&$\PropA$&$\EmptySet$&$\EmptySet$&$\PropA$\\
%             \bottomrule
%        \end{tabular}
%        \caption{A discriminator for $\mciBMatrixName$.}
%\end{table}
With this finite sufficiently expressive nd-\TheB-matrix in hand,
producing a \emph{finite} $\{\PropA,\mciCons\PropA\}$-analytic two-dimensional H-system for it is immediate 
by~\cite[Theorem 2]{greati2021}.
Since \mciName{} inhabits the
$\tAsp$-aspect of $\BEntName{\mciBMatrixName}$,
we may then conclude that:

\begin{theorem}
$\mciName{}$ is axiomatizable by a finite and analytic two-dimensional H-system.
\end{theorem}

Our axiomatization recipe 
delivers an H-system with about 300 rule schemas.
When we simplify it using the streamlining procedures
indicated in that paper,
we obtain a much more succinct and insightful presentation,
with 28 rule schemas,
which we call $\mciBCalcName$.
The full presentation of this system 
is given below:
%can be found
%in Appendix~\ref{app:mcitwod}.
{
\small
\begin{gather*}
    %%%IMP
    \BRuleIn[d]{\PropB}{}{}{\PropA\mciImp\PropB}{\mciBRuleName{\mciImp}{1}}
	\;\;
	\BRuleIn[d]{}{}{}{\PropA,\PropA\mciImp\PropB}{\mciBRuleName{\mciImp}{2}}
	\;
	\BRuleIn[d]{\PropA\mciImp\PropB, \PropA}{}{}{\PropB}{\mciBRuleName{\mciImp}{3}}
	\;
	\BRuleIn[d]{\PropA}{\PropA\mciImp\PropB}{}{\PropB}{\mciBRuleName{\mciImp}{4}}
	\;
	\BRuleIn[d]{\PropA\mciImp\PropB, \mciCons(\PropA\mciImp\PropB)}{}{\PropA\mciImp\PropB}{}{\mciBRuleName{\mciImp}{5}}\\
	%%%AND
	\BRuleIn[d]{\PropA,\PropB}{}{}{\PropA\land\PropB}{\mciBRuleName{\land}{1}}
	\;\;
	\BRuleIn[d]{\PropA\land\PropB}{}{}{\PropA}{\mciBRuleName{\land}{2}}
	\;
	\BRuleIn[d]{\PropA\land\PropB}{}{}{\PropB}{\mciBRuleName{\land}{3}}
	\;
	\BRuleIn[d]{}{\PropA\land\PropB}{}{\PropA\land\PropB}{\mciBRuleName{\land}{4}}
	\;
	\BRuleIn[d]{\PropA{\land}\PropB, \mciCons(\PropA{\land}\PropB)}{}{\PropA{\land}\PropB}{}{\mciBRuleName{\land}{5}}\\
	%%%OR
	\BRuleIn[d]{\PropA}{}{}{\PropA\lor\PropB}{\mciBRuleName{\lor}{1}}
	\;\;
	\BRuleIn[d]{\PropB}{}{}{\PropA\lor\PropB}{\mciBRuleName{\lor}{2}}
	\;\;
	\BRuleIn[d]{\PropA\lor\PropB}{}{}{\PropA,\PropB}{\mciBRuleName{\lor}{3}}
	\;\;
	\BRuleIn[d]{}{\PropA\lor\PropB}{}{\PropA,\PropB}{\mciBRuleName{\lor}{4}}
	\;\;
	\BRuleIn[d]{\PropA{\lor}\PropB, \mciCons(\PropA{\lor}\PropB)}{}{\PropA{\lor}\PropB}{}{\mciBRuleName{\lor}{5}}\\
	%%%CIRC
	\BRuleIn[d]{\mciCons\PropA}{\mciCons\PropA}{}{}{\mciBRuleName{\mciCons}{1}}
	\;\;
	\BRuleIn[d]{}{}{}{\mciCons\mciCons\PropA}{\mciBRuleName{\mciCons}{2}}
	\;\;
	\BRuleIn[d]{}{}{\mciCons\PropA}{\mciCons\PropA}{\mciBRuleName{\mciCons}{3}}
	\;\;
	\BRuleIn[d]{}{\PropA}{}{\mciCons\PropA}{\mciBRuleName{\mciCons}{4}}
	\;\;
	\BRuleIn[d]{}{\mciCons\PropA}{}{\PropA}{\mciBRuleName{\mciCons}{5}}\\
	%%%NEG
	\BRuleIn[d]{}{\mciNeg\PropA,\PropA}{}{}{\mciBRuleName{\mciNeg}{1}}
	\;\;
	\BRuleIn[d]{\mciNeg\PropA,\mciCons\PropA,\PropA}{}{}{}{\mciBRuleName{\mciNeg}{2}}
	\;\;
	\BRuleIn[d]{\mciNeg\PropA,\PropA}{\PropA}{}{}{\mciBRuleName{\mciNeg}{3}}
	\;\;
	\BRuleIn[d]{\mciCons\mciNeg\PropA}{}{\mciNeg\PropA,\PropA}{}{\mciBRuleName{\mciNeg}{4}}\\
	\BRuleIn[d]{}{}{\mciNeg\PropA,\PropA}{\mciNeg\PropA}{\mciBRuleName{\mciNeg}{5}}
	\;\;
	\BRuleIn[d]{}{}{}{\mciNeg\PropA,\mciCons\PropA}{\mciBRuleName{\mciNeg}{6}}
	\;\;
	\BRuleIn[d]{}{}{}{\mciNeg\PropA,\PropA}{\mciBRuleName{\mciNeg}{7}}
	\;\;
	\BRuleIn[d]{}{\PropA}{}{\mciCons\mciNeg\PropA}{\mciBRuleName{\mciNeg}{8}}
\end{gather*}

}

%Theorem~\ref{the:analytic-complete-b-axiomatization}.
%Although the recipe presented in
%Definition~\ref{def:axiomatization} delivers a very complex system 
%(just to have an idea, we obtain 289 rules only by considering 
%the third group of rule schemas),
%the streamlining procedure outlined in
%Section~\ref{sec:simplifying-the-axiomatization} 
%reduces the calculus to only 28 rules,
%consisting thus in a very succint presentation,
%which we denote by $\mciBCalcName$:

%%% removed for space...
Note that the set of rules $\{\mciBRuleName{\ConA}{i} \mid \ConA \in \{\land,\lor,\mciImp\}, i \in \{1,2,3\}\}$ makes it clear that the $\tAsp$-aspect of
the induced \TheB-consequence relation is inhabited by a logic extending positive classical
logic, while the remaining rules for these connectives involve
interactions between the two dimensions.
%Moreover, rule $\mciBRuleName{\mciNeg}{7}$ characterizes this
%logic as \emph{$\mciNeg$-determined}, and 
Also, rule 
$\mciBRuleName{\mciNeg}{2}$ 
indicates that $\mciCons$ satisfies one of the main conditions for being
taken as a consistency connective in the logic inhabiting the $\tAsp$-aspect. 
In fact, all these observations are aligned with the
fact that the logic inhabiting the $\tAsp$-aspect
of $\BConCalcName{\mciBCalcName}$ is precisely \mciName{}.
See, in Figure~\ref{fig:derivationmci},
$\mciBCalcName{}$-derivations showing that, in \mciName{},
$\mciNeg\mciCons\PropA$ and $\PropA{\land}\mciNeg\PropA$
are logically equivalent and
that $\mciCons\mciNeg\mciCons\PropA$ is a theorem.
%(which was given as an axiom in the original \SetFmla{} presentation
%of this logic).}
%(refer to Proposition~\ref{prop:combine-matrices}).

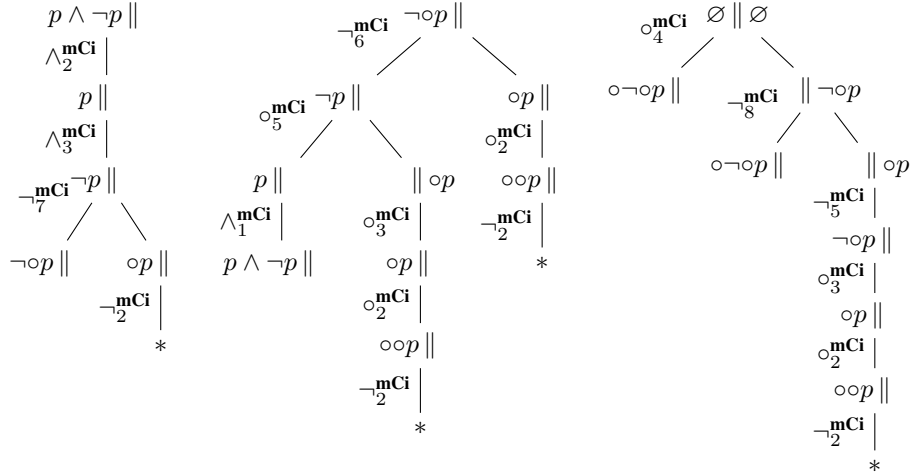
\begin{figure}
    \centering
    \begin{tikzpicture}[every tree node/.style={},
				level distance=1.1cm,sibling distance=-0.2cm,
				edge from parent path={(\tikzparentnode) -- (\tikzchildnode)}, baseline]
		\Tree[.$\BNode{\PropA\land\mciNeg\PropA}{}$
		    \edge[] node[auto=right]{$\mciBRuleName{\land}{2}$};
		    [.$\BNode{\PropA}{}$
		        \edge[] node[auto=right]{$\mciBRuleName{\land}{3}$};
    		    [.$\BNode{\mciNeg\PropA}{}$
    		        \edge[] node[auto=right]{$\mciBRuleName{\mciNeg}{7}$};
    		        [.$\BNode{\mciNeg\mciCons\PropA}{}$
    		        ]
    		        [.$\BNode{\mciCons\PropA}{}$
    		            \edge[] node[auto=right]{$\mciBRuleName{\mciNeg}{2}$};
    		            [.$\StarLabel$
    		            ]
    		        ]
    		    ]
		    ]
		]
	\end{tikzpicture}
    \begin{tikzpicture}[every tree node/.style={},
					level distance=1.1cm,sibling distance=.01cm,
					edge from parent path={(\tikzparentnode) -- (\tikzchildnode)}, baseline]
		\Tree[.$\BNode{\mciNeg\mciCons\PropA}{}$
		    \edge[] node[auto=right]{$\mciBRuleName{\mciNeg}{6}$};
		    [.$\BNode{\mciNeg\PropA}{}$
		        \edge[] node[auto=right]{$\mciBRuleName{\mciCons}{5}$};
		        [.$\BNode{\PropA}{}$
		            \edge[] node[auto=right]{$\mciBRuleName{\land}{1}$};
		            [.$\BNode{\PropA\land\mciNeg\PropA}{}$
		            ]
		        ]
		        [.$\BNode{}{\mciCons\PropA}$
		            \edge[] node[auto=right]{$\mciBRuleName{\mciCons}{3}$};
    		        [.$\BNode{\mciCons\PropA}{}$
    		            \edge[] node[auto=right]{$\mciBRuleName{\mciCons}{2}$};
        		        [.$\BNode{\mciCons\mciCons\PropA}{}$
        		            \edge[] node[auto=right]{$\mciBRuleName{\mciNeg}{2}$};
        		            [.$\StarLabel$
        		            ]
        		        ]
    		        ]
		        ]
		    ]
		    [.$\BNode{\mciCons\PropA}{}$
		        \edge[] node[auto=right]{$\mciBRuleName{\mciCons}{2}$};
		        [.$\BNode{\mciCons\mciCons\PropA}{}$
		            \edge[] node[auto=right]{$\mciBRuleName{\mciNeg}{2}$};
    		        [.$\StarLabel$
    		        ]
		        ]
		    ]
		]
	\end{tikzpicture}
    \begin{tikzpicture}[every tree node/.style={},
						level distance=1.0cm,sibling distance=-.2cm,
						edge from parent path={(\tikzparentnode) -- (\tikzchildnode)}, baseline]
			\Tree[.$\BNode{\EmptySet}{\EmptySet}$
			   \edge[] node[auto=right]{$\mciBRuleName{\mciCons}{4}$};  [.$\BNode{\mciCons\mciNeg\mciCons\PropA}{}$
			    ]		    
			    [.$\BNode{}{\mciNeg\mciCons\PropA}$
    			    \edge[] node[auto=right]{$\mciBRuleName{\mciNeg}{8}$};
			    	[.$\BNode{\mciCons\mciNeg\mciCons\PropA}{}$
			        ]
			        [.$\BNode{}{\mciCons\PropA}$
			            \edge[] node[auto=right]{$\mciBRuleName{\mciNeg}{5}$};
			            [.$\BNode{\mciNeg\mciCons\PropA}{}$
    			            \edge[] node[auto=right]{$\mciBRuleName{\mciCons}{3}$};
    			            [.$\BNode{\mciCons\PropA}{}$
    			                \edge[] node[auto=right]{$\mciBRuleName{\mciCons}{2}$};
        			            [.$\BNode{\mciCons\mciCons\PropA}{}$
        			                \edge[] node[auto=right]{$\mciBRuleName{\mciNeg}{2}$};
        			                [.$\StarLabel$
        			                ]
        			            ]
    			            ]
			            ]
			        ]
			    ]
			]
	\end{tikzpicture}
    \caption{$\mciBCalcName$-derivations showing, respectively, that 
    $\BConCalc{\PropA\land\mciNeg\PropA}{\EmptySet}{\EmptySet}{\mciNeg\mciCons\PropA}{\mciBCalcName}$,
    $\BConCalc{\mciNeg\mciCons\PropA}{\EmptySet}{\EmptySet}{\PropA\land\mciNeg\PropA}{\mciBCalcName}$ and
    $\BConCalc{\EmptySet}{\EmptySet}{\EmptySet}{\mciCons\mciNeg\mciCons\PropA}{\mciBCalcName}$. Note that, for a
    cleaner presentation, we omit the formulas inherited from parent nodes.}
    \label{fig:derivationmci}
\end{figure}

\section{Concluding remarks}
\vspace{-.5em}
In this work, we introduced a mechanism 
for combining two non-deterministic logical matrices
into a non-deterministic \TheB-matrix,
creating the possibility of producing
finite and analytic two-dimensional axiomatizations
for one-dimensional logics that may fail to be finitely
axiomatizable in terms of one-dimensional Hilbert-style systems.
It is worth mentioning that, as proved in~\cite{greati2021}, one may perform proof search and
countermodel search over the resulting two-dimensional systems
in time at most exponential
on the size of the \TheB-statement of interest
through a straightforward proof-search algorithm.

%In the present work, we introduced a novel technique for
%combining two distinct one-dimensional logics specified by non-deterministic matrices into
%a single two-dimensional logic. 
%Each of the former logics lives in
%an independent dimension of the combined logic
%and can be easily recovered just by considering the
%consecutions solely involving the corresponding dimension.
%The main advantage in performing this combination
%is that the resulting two-dimensional logic is potentially
%more expressive than the original one-dimensional logics.
%We proved this by unveiling the possibility of producing finite two-dimensional axiomatizations
%for one-dimensional logics that are not finitely
%axiomatizable in terms of one-dimensional Hilbert-style systems.
%In case the \TheB-matrix resulting from the
%combination of the non-deterministic matrices
%satisfy a criterion of sufficient expressiveness,
%one can obtain algorithmically a corresponding two-dimensional
%axiomatization that is not only finite
%when the \TheB-matrix is finite, but is 
%also \emph{analytic}. When these two property are
%verified, one can perform proof search and
%countermodel search in time at most exponential
%in the size of the statement of interest
%through a straightforward algorithm, as proved in~\cite{greati2021}.

We illustrated the above-mentioned combination mechanism with two examples,
one of them corresponding to a well-known logic of formal
inconsistency called \mciName{}.
We ended up proving not only that this logic
is not finitely axiomatizable in one dimension, but also
that it is the limit of a strictly increasing
chain of LFIs extending the logic \mbcName{}.
From the perspective of the study of
        \TheB-consequence relations,
        these examples allow us to eliminate the suspicion
        that a two-dimensional H-system $\BCalcA$ may always be converted into
        \SetSet{} H-systems for the logics inhabiting the one-dimensional aspects of $\BConCalcName{\BCalcA}$ without losing any
        desirable property (in this case, finiteness of
        the presentation).
%        	In other words, there is no procedure, in general, to
%        	produce a finite one-dimensional H-system
%        	for the logic living in the $\tAsp$-aspect
%        	or in the $\fAsp$-aspect of a \TheB-consequence
%        	relation
%        	from a finite two-dimensional H-system
%        	that axiomatizes this same \TheB-consequence relation.

At first sight, the formalism of two-dimensional
H-systems may be confused with
the formalism of $n$-sided sequents~\cite{avron2007,avron2005igpl},
in which the objects manipulated by rules
of inference (the so-called \emph{$n$-sequents}) accommodate more than
two sets of formulas in their structures.
The reader interested in
a comparison between these two different approaches is referred to
the concluding remarks of
\cite{greati2021}.

We close with some observations 
regarding $\mciBMatrixName$ and the two-dimensional H-system $\mciBCalcName$.
%(refer to Appendix~\ref{app:mcitwod}).
A one-dimensional logic $\SSHCR{}$ is said to be
\emph{$\mciNeg$-consistent} when
$\FmA,\mciNeg\FmA \SSHCR{} \EmptySet$
and \emph{$\mciNeg$-determined} when
$\EmptySet \SSHCR{} \FmA,\mciNeg\FmA$
for all $\FmA \in \LangSet{\SigA}{\PropSetA}$.
A \TheB-consequence relation
$\BConName{}$
is said to \emph{allow for gappy reasoning}
when $\nBCon[]{\FmA}{}{\FmA}{}$
and to \emph{allow for glutty reasoning} when
$\nBCon[]{}{\FmA}{}{\FmA}$,
for some $\FmA \in \LangSet{\SigA}{\PropSetA}$.
Notice that $\mciNeg$-determinedness 
in the logic inhabiting the $\tAsp$-aspect of a
\TheB-consequence relation
by no
means implies the disallowance of gappy reasoning in the two-dimensional
setting: we still have $\mciF \in \ValueSetComp{\DesSetFive} \cap \ValueSetComp{\RejSetFive}$,
so one may both non-accept and non-reject a formula $\FmA$ in $\BConCalcName{\mciBCalcName}$,
even though non-accepting both $\FmA$ and its negation in \mciName{} 
is not
possible, in view of rule $\mciBRuleName{\mciNeg}{7}$. 
Similarly, the recovery of $\mciNeg$-consistency achieved via
$\mciCons$ in such logic 
does not coincide with the gentle disallowance of glutty reasoning in
$\BConCalcName{\mciBCalcName}$,
that is, we do not have, in general, 
$\BConCalc[]{\PropA,\mciCons\PropA}{}{\PropA}{}{\mciBCalcName}$
or
$\BConCalc[]{\PropA}{}{\mciCons\PropA,\PropA}{}{\mciBCalcName}$,
even though for binary compounds both are derivable in view of rules
$\mciBRuleName{\ConA}{5}$, for $\ConA \in \{\land,\lor,\mciImp\}$,
and $\mciCons_1^{\mciName{}}$.
With these observations we hope to call attention to
the fact that \TheB-consequence relations
open the doors for further developments 
concerning the study of paraconsistency (and, dually, of paracompleteness), as well as the~study of 
recovery operators~\cite{carnielli2019}.
%In the previous section, we have briefly discussed some
%aspects of the two-dimensional H-system $\mciBCalcName$
%highlighting the possibility of studying
%the recovery of classical reasoning 
%in two-dimensional logics, which, as we 
%proved with the case of \mciName{}, does not
%necessarily coincide with the recovery of $\mciNeg$-consistency
%in one of the dimensions. There is, therefore, a path
%of further research on \emph{two-dimensional} logics of
%formal inconsistency (and, dually, of undeterminedness), which
%also involves understanding the characteristics
%of a classical negation in this more general setting.

\section*{Acknowledgements}

V.\ Greati acknowledges support from
%the Coordenação de Aperfeiçoamento de Pessoal de Nível Superior ---
CAPES
--- Finance Code 001
and from the
FWF project P 33548. J.~Marcos acknowledges support from
%Conselho Nacional de Desenvolvimento Científico e Tecnológico (CNPq). 
CNPq.

%
% ---- Bibliography ----
%
% BibTeX users should specify bibliography style 'splncs04'.
% References will then be sorted and formatted in the correct style.
%
% \bibliographystyle{splncs04}
% \bibliography{mybibliography}
%
\bibliographystyle{plain}
\bibliography{ms.bib}

\begin{appendix}
    \section{Detailed proofs}
\label{app:proofs}

%\begin{lemma}
\textbf{Lemma~\ref{lem:mcisequence}.}
\textit{
    For each $1 \leq k < \omega$,
    let 
    $\SetFmlaCR{k} \;\SymbDef\; 
    \SetFmlaCR{\mciHilbertName^{2k - 1}}$.
    Then~$\SetFmlaCR{1} \;\subseteq\; \SetFmlaCR{2} \;\subseteq \ldots$,
    and
    \begin{equation*}
    	\SetFmlaCR{\mciName} \;= \Supremum{}_{1 \leq k < \omega} \SetFmlaCR{k}.
    \end{equation*}
}
%\end{lemma}
\begin{proof}
By Definition~\ref{def:mcik}, every rule schema in
$\mciHilbertName^{2k - 1}$
is also in
$\mciHilbertName^{2(k+1) - 1}$,
thus, for every $1 \leq k < \omega$,
we have
$\SetFmlaCR{k} \;\subseteq\; \SetFmlaCR{k+1}$.
Let~$\SetFmlaCR{\omega} \;\SymbDef\; \Supremum{}_{1 \leq k < \omega} \SetFmlaCR{k}.$
From right to left, 
if $\FmSetA \SetFmlaCR{\omega} \FmB$, then,
	for every Tarskian consequence relation $\SetFmlaCR{\ast}$ over $\SigMCI$ 
	such that $\SetFmlaCR{\ast}\; \supseteq\; \SetFmlaCR{k}$
	for all $k \in \omega$, we have
	$\FmSetA \SetFmlaCR{\ast} \FmB$.
	By Proposition~\ref{prop:mciextendsall},
	then, we have
$\FmSetA\SetFmlaCR{\mciName}\FmB$, in particular.
From left to right,
suppose that $\FmSetA\SetFmlaCR{\mciName}\FmB$
and %let $\FmA_1,\ldots,\FmA_n$
%be 
consider a derivation bearing witness to this fact.
Let $m \in \omega$ be such that only instances
of the rule schemas
\ref{rule:mciCij}, for $0 \leq j \leq m$,
and possibly instances of the other rule schemas
not of the form of \ref{rule:mciCij}
are applied in that derivation.
	Let
	$\SetFmlaCR{\ast\ast}$ be a
	Tarskian consequence relation over $\SigMCI$ 
	such that
	$\SetFmlaCR{\ast\ast} \;\supseteq\; \SetFmlaCR{k}$
	for all $1 \leq k < \omega$.
	Then, in particular,
	$\SetFmlaCR{\ast\ast}\; \supseteq\; \SetFmlaCR{m} \;=\; \SetFmlaCR{\mciHilbertName^{2m-1}}$.
	Since all schemas
	\ref{rule:mciCij}, for $0 \leq j \leq m$,
	are in $\mciHilbertName^{2m-1}$,
	we have $\FmSetA \SetFmlaCR{\mciHilbertName^{2m-1}} \FmB$
	and then
	$\FmSetA \SetFmlaCR{\ast\ast} \FmB$.
	As $\SetFmlaCR{\ast\ast}$ was arbitrary, we are done.\qed
\end{proof}

\noindent\textbf{Lemma 
\ref{lem:itneg}. }
For all $k \geq 1$ and $1 \leq m \leq 2k$,
\[
\AlgInterp{\mciNeg}{\AlgA_k}^m(\Suc{k}+1) =
\begin{cases}
(\Suc{k}+1)+\frac{m}{2}, & \text{if } m \text{ is even}\\
1 +\frac{m+1}{2}, & otherwise\\
\end{cases}
\]
%where $\mciNeg$ is $\AlgInterp{\mciNeg}{\AlgA_k}$.
\begin{proof}
Let $k \geq 1$. We prove the lemma by strong induction on $1 \leq m \leq 2k$. For $m=1$,
we have $\mciNeg(\Suc{k}+1)=(\Suc{k}+1)-(\Suc{k}-1) = 2 = 1 + \frac{1+1}{2}$.
	Assume now that (IH): the present lemma holds for all $m' < m$, for a given $m > 1$.
\begin{itemize}
    \item Suppose that $m = 2s$, with $1 \leq s \leq k$.
	    By (IH), we have that
    $\mciNeg^{2s}(\Suc{k}+1) = \mciNeg(\mciNeg^{2s-1}(\Suc{k}+1)) = \mciNeg(1+\frac{(2s-1)+1}{2}) = \mciNeg(1+s)$.
    By the interpretation of $\mciNeg$, as $2 \leq 1+s \leq \Suc{k}$, we have $\mciNeg(1+s) = 1+s+\Suc{k} = (\Suc{k}+1) + \frac{m}{2}$.
    \item Suppose that $m = 2s+1$, with $1 \leq s \leq k-1$.
	    By (IH), we have
    $\mciNeg^{2s+1}(\Suc{k}+1) = \mciNeg(\mciNeg^{2s}(\Suc{k}+1)) = \mciNeg(\Suc{k}+1 + \frac{2s}{2}) = \mciNeg(\Suc{k}+1+s)$. As $\Suc{k}+2 \leq \Suc{k}+1+s \leq \Suc{k}+k$, the interpretation of $\mciNeg$ gives
    us that $\mciNeg(\Suc{k}+1+s) = (\Suc{k}+1+s) - (\Suc{k}-1) = s + 2 = \frac{m-1}{2} + 2 = (\frac{m-1}{2}+1)+1 = 1+\frac{m+1}{2}$.\qed%\qedhere
\end{itemize}
\end{proof}

\noindent \textbf{Lemma 
\ref{lem:mcistrictinc}. }
For all $1 \leq k < \omega$, 
we have
$\SetFmlaCR{\mciHilbertName^{2\Suc{k}-1}} \mciCons\mciNeg^{2k}\mciCons\PropA$
but 
$\not\SetFmlaCR{\mciHilbertName^{2k-1}} \mciCons\mciNeg^{2k}\mciCons\PropA$.
%where $k^\ast \SymbDef k + 1$.
\begin{proof}
Let $1 \leq k < \omega$.
We start by showing that
$\mciHilbertName^{2k-1}$ is sound for
$\MatA_k$.
The rule of inference from positive classical logic
are sound with respect to $\MatA_k$, since
the mapping $h$ given by $h(\ValA) = \TheFalse$
if $\ValA \in \{1,\ldots,\Suc{k}\}$
and $h(\ValA) = \TheTrue$ otherwise
is a strong homomorphism from the positive fragment of $\MatA_k$
onto $\BooleanMat$,
the usual two-valued
matrix that determines positive classical logic.
%%%
%%%
%\begin{addmargin}[4em]{0em}
%\end{addmargin}
%
Below we show
soundness of the remaining rules (all of which are axiom schemas),
which involve the
connectives $\mciNeg$ and $\mciCons$. 
The lemma just proved will be employed in the
case of~\ref{rule:mciCij}.

\begin{description}[font=\textrm]
	\item[\ref{rule:mciExcMid}] Suppose that
    $v(\FmA \lor \mciNeg\FmA) \in \ValueSetComp{D_k}$,
    then $v(\FmA) \in \ValueSetComp{D_k}$
    and $v(\mciNeg\FmA) \in \ValueSetComp{D_k}$. From
    the latter, we have $\Suc{k}+1 \leq v(\FmA) \leq 2\Suc{k}-1$, but then $v(\FmA) \in D_k$,
    a contradiction.
    \item[\ref{rule:mcibc}] \sloppy Suppose that
    $v(\mciCons\FmA \mciImp (\FmA \mciImp (\mciNeg\FmA \mciImp \FmB))) \in \ValueSetComp{D_k}$.
    Then (a): $v(\mciCons\FmA) \in D_k$ and
    $v(\FmA \mciImp (\mciNeg\FmA \mciImp \FmB)) \in \ValueSetComp{D_k}$.
    From the latter, reasoning in the same
    way, we have (b): $v(\FmA) \in D_k$,
    (c): $v(\mciNeg\FmA) \in D_k$
    and $v(\FmB) \in \ValueSetComp{D_k}$.
    From (b), (c) and the interpretation
    of $\mciNeg$, we have that $v(\FmA)=2\Suc{k}$,
    but then $v(\mciCons\FmA) = 1 \in \ValueSetComp{D_k}$, contradicting (a).
    \item[\ref{rule:mciCi}] Suppose that
    $v(\mciNeg\mciCons\FmA \mciImp (\FmA \land \mciNeg\FmA)) \in \ValueSetComp{D_k}$.
    Then (a): $v(\mciNeg\mciCons\FmA) \in D_k$
    and (b): $v(\FmA \land \mciNeg\FmA) \in \ValueSetComp{D_k}$.
    From (a), we have (c): $1 \leq v(\mciCons\FmA) \leq \Suc{k}$ or $v(\mciCons\FmA) = 2\Suc{k}$.
    From (b),
    we have that either (b1): $v(\FmA) \in \ValueSetComp{D_k}$ or (b2): $v(\mciNeg\FmA) \in \ValueSetComp{D_k}$. By cases:
    \begin{itemize}
        \item if (b1), then $v(\mciCons\FmA) = \Suc{k} + 1$, contradicting (c).
        \item if (b2), then $\Suc{k}+1 \leq v(\FmA) \leq 2\Suc{k}-1$ by the interpretation
        of $\mciNeg$, but then
        $v(\mciCons\FmA) = \Suc{k}+1$
        by the interpretation of $\mciCons$,
        contradicting (c).
    \end{itemize}
    %\item[\ref{rule:mciCiZero}]
    \item[\ref{rule:mciCij}]
    For $j = 0$, suppose that $v(\mciCons\mciCons\FmA) \in \ValueSetComp{D_k}$. Then,
    $v(\mciCons\FmA) = 2\Suc{k}$, which
    is impossible from the interpretation
    of $\mciCons$.
     Let $1 \leq j \leq 2k-1$.
    Suppose that $v(\mciCons\mciNeg^j\mciCons\FmA) \in \ValueSetComp{D_k}$. Then, by the
    interpretation of $\mciCons$, we have
    (a): $v(\mciNeg^j\mciCons\FmA) = 2\Suc{k}$. By cases
    on the possible values of $v(\mciCons\FmA)$:
    \begin{itemize}
        \item if $v(\mciCons\FmA) = \Suc{k}+1$:
        by Lemma \ref{lem:itneg}, if $j$ is even,
	we have $\mciNeg^j(\Suc{k}+1) = (\Suc{k}+1) + \frac{j}{2}
		    = (\Suc{k}+1)+s = k+2 + s \leq k+2+k-1 = 2k + 1 < 2\Suc{k}$, with $1 \leq s \leq k - 1$.
	If $j$ is odd, then
		    $\mciNeg^j(\Suc{k}+1) = 1 + \frac{j+1}{2} = 1 + \frac{2s-1+1}{2} = 1 + s \leq 1 + k < 2\Suc{k}$, with $1 \leq s \leq k$.
		    Both cases contradict (a).
%	    \item if $v(\mciCons\FmA) = 1$, we proceed by induction on $j$: for $j = 1$, 
%		    we have $v(\mciNeg\mciCons\FmA) = \Suc{k}+1 \in D$.
%		    Assume that $v(\mciNeg^{j'}\mciCons\FmA) \in D$.
%		    Then $2 \leq \mciNeg v(\mciNeg^{j'}\mciCons\FmA) \leq \Suc{k}+1$.
%		    Therefore, for each $1 \leq j \leq 2k+1$, 
%		    $v(\mciNeg^j\mciCons\FmA) \neq 2\Suc{k}$,
%		    and thus
%		    $v(\mciCons\mciNeg^j\mciCons\FmA) = \Suc{k}+1 \in D$.
%        since $j \leq 2k-1$, we would have
%        $\mciNeg^j(\Suc{k}+1) \leq \mciNeg^{2k-1}(\Suc{k}+1) = (\Suc{k}+1)+\frac{2k-1}{2} < 2\Suc{k}$.
%        If $j$ is odd, then
%        $\mciNeg^j(\Suc{k}+1) \leq \mciNeg^{2k-1}(\Suc{k}+1) = 1+k = \Suc{k} < 2\Suc{k}$.
%        Therefore, $v(\mciNeg^j\mciCons\FmA) \neq 2\Suc{k}$,
%        contradicting (a).
        \item if $v(\mciCons\FmA) = 1$:
        may apply essentially the same reasoning as in the previous
		    case, since $v(\mciNeg^j\mciCons\FmA) = \mciNeg^{j-1}v(\mciNeg\mciCons\FmA) = \mciNeg^{j-1}(\Suc{k}+1)$.
    \end{itemize}
\end{description}

\noindent For the second part of the proof, take a $\MatA_k$-valuation
$\ValuationA$ such that
$\ValuationA(\PropA) = 1$. Then $\ValuationA(\mciCons \PropA) = \Suc{k}+1$
and, since $2k$ is even, by Lemma~\ref{lem:itneg},
we have $\mciNeg^{2k}(\Suc{k}+1) = (\Suc{k}+1)+\frac{2k}{2} = \Suc{k}+1 + k = 2k+2 = 2\Suc{k}$.
Thus $\ValuationA(\mciNeg^{2k}\mciCons \PropA) = 2\Suc{k}$
and, by the interpretation of $\mciCons$,
we have $v(\mciCons\mciNeg^{2k}\mciCons \PropA) = 1 \in \ValueSetComp{D_k}$, and
we are done.\qed
\end{proof}

\noindent\textbf{\autoref{prop:combine-matrices}.}
Let $\MatA \SymbDef \Struct{\AlgA,\DesSetA}$
be a $\SigA$-nd-matrix and suppose that $U \subseteq \ValSetA\times\ValSetA$ contains
all and only the pairs
of distinct truth-values
that fail to be
separable in $\MatA$.
If, for some $\MatA' \SymbDef \Struct{\AlgA,\DesSetA'}$,
the pairs in $U$ are separable in $\MatA'$, then
 $\MatA \MatBProdFn \MatA'$
is sufficiently expressive
(thus, axiomatizable by an analytic \SetTSetT{} H-system, that is finite whenever $\AlgA$ is finite).
\begin{proof}
Let $\Pair{\ValC}{\ValD} \in \ValSetA\times\ValSetA$.
In case 
$\Pair{\ValC}{\ValD} \not\in U$, there is a separator
$\SepA$ for $\Pair{\ValC}{\ValD}$
in $\MatA$,
that is,
$\Separated{\AlgInterp{\SepA}{\AlgA}(\ValC)}{\AlgInterp{\SepA}{\AlgA}(\ValD)}{\DesSetA}$.
Otherwise, if all pairs in $S$
are
separable in $\MatA'$, then,
in particular, $\Pair{\ValC}{\ValD}$
is also separable in $\MatA'$, say, by a separator
$\SepA'$, that is,
$\Separated{\AlgInterp{\SepA^\prime}{\AlgA}(\ValC)}{\AlgInterp{\SepA^\prime}{\AlgA}(\ValD)}{\DesSetA'}$.
Therefore, every pair of truth-values of $\AlgA$ is separable
in $\MatA \MatBProdFn \MatA'$, and so the latter is sufficiently expressive.
By the procedure in~\cite{greati2021}, $\MatA \MatBProdFn \MatA'$
is axiomatizable by an analytic \SetTSetT{} H-system
that is finite if $\AlgA$ is finite.
\end{proof}

%\section{Finite and analytic two-dimensional system for \mciName{}}
%\label{app:mcitwod}
%\input{tex/mcisystem}
\end{appendix}

\end{document}